\documentclass{article}

\usepackage{microtype}
\usepackage{graphicx}
\usepackage{subcaption}
\usepackage{booktabs} 

\usepackage{hyperref}

\usepackage[accepted]{icml2023}

\usepackage{amsmath}
\usepackage{amssymb}
\usepackage{mathtools}
\usepackage{amsthm}

\usepackage[capitalize,noabbrev]{cleveref}

\usepackage{thm-restate}
\usepackage{bm}

\theoremstyle{plain}
\newtheorem{theorem}{Theorem}[section]

\newtheorem{lemma}[theorem]{Lemma}

\theoremstyle{definition}
\newtheorem{definition}[theorem]{Definition}

\theoremstyle{remark}
\newtheorem{remark}[theorem]{Remark}

\usepackage[textsize=tiny]{todonotes}

\usepackage{algorithm}
\usepackage{algpseudocode}
\usepackage{wrapfig}

\DeclareMathOperator*{\E}{\mathbb{E}}

\DeclareMathOperator*{\argmin}{arg\,min}
\newcommand{\eps}{\varepsilon}

\newcommand{\R}{\mathbb{R}}

\newcommand{\allpairs}{\textit{nAllPairs}}
\newcommand{\fast}{\textit{fastConst}}

\newcommand{\sandeep}[1]{\textcolor{orange}{[sandeep\@ifnotempty{#1}{: #1}]}}

\newcommand{\shyam}[1]{\textcolor{violet}{[shyam\@ifnotempty{#1}{: #1}]}}

\icmltitlerunning{\hfill Data Structures for Density Estimation \hfill\thepage}

\begin{document}

\twocolumn[
\icmltitle{Data Structures for Density Estimation}



\icmlsetsymbol{equal}{*}

\begin{icmlauthorlist}
\icmlauthor{Anders Aamand}{equal,MIT}
\icmlauthor{Alexandr Andoni}{Col}
\icmlauthor{Justin Y. Chen}{MIT}
\icmlauthor{Piotr Indyk}{MIT}
\icmlauthor{Shyam Narayanan}{MIT}
\icmlauthor{Sandeep Silwal}{MIT}
\end{icmlauthorlist}

\icmlaffiliation{MIT}{Computer Science and Artificial Intelligence Laboratory, Massachusetts Institute of Technology, Cambridge, MA 02139, USA.}
\icmlaffiliation{Col}{Data Science Institute, Columbia University, New York, NY 10027, USA}

\icmlcorrespondingauthor{Anders Aamand}{aamand@mit.edu}
\icmlcorrespondingauthor{Alexandr Andoni}{andoni@mit.edu}
\icmlcorrespondingauthor{Justin Y. Chen}{justc@mit.edu}
\icmlcorrespondingauthor{Piotr Indyk}{indyk@mit.edu}
\icmlcorrespondingauthor{Shyam Narayanan}{shyamsn@mit.edu}
\icmlcorrespondingauthor{Sandeep Silwal}{silwal@mit.edu}

\icmlkeywords{Machine Learning, ICML}

\vskip 0.3in
]



\printAffiliationsAndNotice{$^*$Authors listed alphabetically.} 

\begin{abstract}
We study statistical/computational tradeoffs for the following density estimation problem: given $k$ distributions $v_1, \ldots, v_k$ over a discrete domain of size $n$, and sampling access to a distribution $p$, identify $v_i$ that is ``close'' to $p$. Our main result is the first data structure that, given a sublinear (in $n$)  number of samples from $p$, identifies  $v_i$ in time sublinear in $k$. We also give an improved version of the algorithm of~\cite{acharya2018maximum} that reports $v_i$ in time linear in $k$. The experimental evaluation of the latter algorithm shows that it achieves a significant reduction in the number of operations needed to achieve a given accuracy compared to prior work. 
\end{abstract}

\section{Introduction}

Finding the hypothesis that best matches a given set of samples, i.e., the {\em density estimation problem}, is a fundamental task in statistics and machine learning. In the finite setting, this task can be formulated as follows: given $k$ distributions $v_1, \ldots, v_k$ over some domain $U$ and access to samples from a distribution $p$ over $U$, output $v_i$ that is ``close'' to $p$. 
In the ``proper'' case, we assume that $p=v_j$ for some $1 \le j \le k$, and the goal is to output $v_i$ such that $\|p-v_i\|_1 \le \eps$ for some error parameter $\eps>0$. In the more general ``improper'' case, $p$ is arbitrary and the goal is to report $v_i$ such that
\[\|p-v_i\|_1 \le C \cdot \min_j \|p- v_j\|_1 + \eps \]
for some constant $C>1$ and error parameter $\eps>0$. 

The density estimation problem has been studied extensively. The work of Scheffe~\cite{Scheffe1947useful}, Yatracos~\cite{yatracos1985rates}, Devroye-Lugosi~\cite{devroye2001combinatorial}  showed that $O(\log(k)/\eps^2)$ samples are sufficient to solve the improper version of the problem for some constant $C>1$ with a constant probability in $O(k^2 \log (k)/\eps^2 )$ time.   Recently Acharya et al~\cite{acharya2018maximum} improved the time to $O(k \log (k)/\eps^2 )$ time.

Despite the generality of the formulation, this approach yields almost optimal sampling bounds for many natural classes of distributions such as mixtures of Gaussians~\cite{daskalakis2014faster,suresh2014spherical,diakonikolas2019robust}, see also the survey~\cite{diakonikolas2016learning}. Furthermore, the method has been extended to enable various forms of privacy~\cite{canonne2019structure,bun2019private,gopi2020locally,kamath2020private}.

When the domain $U$ is a discrete set of size $n$, the problem can be viewed as a variant of approximate nearest neighbor search over $n$-dimensional vectors under the $L_1$ norm, a problem that has been studied extensively~\cite{andoni2018approximate}. Specifically, given a set of $k$ distributions $v_1, \ldots, v_k$, the goal of $C$-approximate nearest neighbor search is to build a data structure that, given any query distribution $p$, returns $v_i$ such that $\|p-v_i\|_1 \le C \cdot \min_j \|p- v_j\|_1$.

However, in the density estimation problem, the query distribution $p$ is not specified fully, but instead the algorithm is given only samples from $p$. This makes the problem substantially richer and more complex, as in addition to the usual {\em computational} resources (data structure space, query time), one also needs to consider the number of samples taken from $p$, a {\em statistical} resource.  Thus, designing data structures for this problem involves making tradeoffs between the computational and statistical resources, with the known algorithms forming the endpoints of the tradeoff curve. In particular:
\begin{itemize}
\item If the goal is to optimize the {\em computational} efficiency of the data structure, one can learn the query distribution $p$ up to an additive error of $\eps$ using $O(n/\eps^2)$ samples ~\cite{kamath2015learning}, and use any $c$-approximate nearest neighbor algorithm for the $L_1$ norm. In particular, the algorithm of~\cite{andoni2015optimal} yields (roughly) $O(nk+k^{1+1/(2c-1)})$ space and $O(nk^{1/(2c-1)})$ query time. This leads to an algorithm with polynomial space and sublinear (in $k$) query time, but at the price of using a linear (in $n$) number of samples.
\item If the goal is to optimize the {\em statistical} efficiency of the data structure, then the aforementioned result of \cite{acharya2018maximum} yields a data structure that uses only a logarithmic number of samples, but at the price of near linear (in $k$) query time. 
\end{itemize}

These two point-wise results raise the question of whether ``best of both worlds'' data structures exists, i.e., whether there are data structures which are efficient in both statistical and computational terms. This is the focus of the paper.

\paragraph{Our results} In this paper we initiate the study of computational/statistical tradeoffs of the (discrete) density estimation problem. Our main theoretical contributions are as follows:
\begin{itemize} 
\item We present the first data structure that solves the {\em proper} version of the problem with polynomial  space, sublinear query time {\em and} sublinear sampling complexity. This demonstrates that one can achieve non-trivial complexity bounds on all three parameters (space, query time, sampling complexity) {\em simultaneously}.
\item We also present an improved version of the data structure from~\cite{acharya2018maximum} for the {\em improper} case, reducing its query time from $O(k \log (k)/\eps^2)$ to $O(k/\eps^2)$, i.e., linear in $k$, while retaining its optimal sampling complexity bound of $O(\log(k)/\eps^2)$. 
\end{itemize}

On the empirical side, we experimentally evaluate the linear-time algorithm and compare its performance to the \cite{acharya2018maximum} algorithm on synthetic and on real networking data. These experiments display the practical benefits of the faster algorithm.  For example, for synthetic data, we demonstrate that our faster algorithm achieves over \textbf{2x} reduction in the number of comparisons needed to achieve the same level of accuracy. Similarly, our experiments on network data show up to \textbf{5x} reduction in the number of comparisons.

\paragraph{Open questions} We view a major part of the contribution of this paper as initiating the study of statistical/computational tradeoffs in data structures for density estimation. We design a data structure with polynomial space and sublinear query and sampling complexity but only in the proper case and only achieving query and sampling bounds that are slightly sublinear in the input parameters. \emph{Extending these results (1) to the improper case, (2) with stronger sublinear bounds, or (3) showing lower bounds are all exciting open problems.}

\subsection{Preliminaries}\label{sec:prelim}
We assume the discrete distributions $v_1, \ldots, v_k$ over the domain $[n]$ are fully specified. We assume $k$, the number of distributions, is much larger than $n$ but still polynomially related to $n$. For instance, $n^{1.01} \le k \le n^{C}$ for $C=100$ suffices. This mimics the classic nearest neighbor search (NNS) literature setting, where the aim is to get algorithms sublinear in the dataset size, which exactly corresponds to $k$ in our setting. 

Let $p$ be the unknown distribution we sample from. We actually attain results stronger than the truly proper case: we relax the equality condition and assume that there exists $v^*$ among the $v_i$'s such that $\|p - v^*\|_2 \le \frac{\eps}{2 \sqrt{n}}$.

We use the standard Poissonization trick and take $s' = \text{Pois}(s)$ samples \footnote{$\text{Pois}(s)$ denotes a draw from a Poisson distribution with parameter $s$.}. This doesn't affect the sample complexity asymptotically as $s' = \Theta(s)$ with high probability.  This is a standard trick used in distribution testing \cite{canonne2020survey, szpankowski2011average, valiant2011power} and ensures that for all $i \in [n]$, the number of samples observed that are equal to $i$ when sampling from $p$ is distributed as $\text{Pois}(s \cdot p(i))$ and the counts are independent.
We say $b \lesssim a$ for $a,b > 0$ if $b \le C \cdot a$ for some absolute positive constant $C$ which does not depend on $a$ or $b$.

We recall the guarantees of $\ell_{\infty}$ and $\ell_2$ NNS data structures.

\begin{theorem}{\cite{indyk2001approximate}}\label{thm:l_infinity_nns}
There exists some absolute constant $C$ such that given a dataset $X \subset \R^n$ with $|X| = k$, for every $\delta \in (0,1)$, we can solve the $\frac{C}{\delta} \log \log n$-approximate nearest neighbor problem on $X$ in $\ell_{\infty}$ using $\tilde{O}(nk^{1+\delta})$ space and $\tilde{O}(n)$ query time.
\end{theorem}

The following theorem is obtained by using a standard combination of randomized dimensionality reduction theorem of Johnson-Lindenstrauss and locality-sensitive hashing~\cite{har2012approximate}.

\begin{theorem}{\cite{har2012approximate}}\label{thm:l2_nns}
Given a dataset $X \subset \R^n$ with $|X| = k$, for every $c > 1$, we can solve the $c$-approximate nearest neighbor problem on $X$ in $\ell_{2}$ using $(k^{1+\rho} + kn) \log^{O(1)} k$ space and $(k^{\rho} + n) \log^{O(1)}k$ query time where $\rho \le 1/c$.
\end{theorem}
Lastly, we define some convenient notation used shortly.
\begin{definition}\label{def:vector_A}
For a vector $x \in \mathbb{R}^n$ and a set $A \subseteq [n]$, let $x_A \in \R^m$ denote the vector which is equal to $x$ in coordinates of set $A$ and zero otherwise. $x(i)$ refers to the $i$th coordinate of the vector $x$. 
\end{definition}

\section{Motivating our algorithm}\label{sec:motivation}

As we discuss in the next section, our algorithm works by considering the empirical sample distribution $\hat p$ (where $\hat p_i=X(i)/s$ if $i$ is sampled $X(i)$ times) and delicately combines nearest neighbor search data structures both for $\ell_2$ and $\ell_\infty$ distance. 
Before describing our algorithm, we first explain why simpler approaches do not work. First, since we are interested in finding a close distribution in $\ell_1$ distance, it is a natural idea to simply construct an approximate $\ell_1$ nearest neighbor data structure on the known distributions and directly querying it for $\hat p$. As the following lemma illustrates, this approach fails to return a close distribution even with access to an \emph{exact} nearest neighbor data structure which returns a true nearest neighbour to the empirical distribution. Surprisingly, it can even be \emph{much worse} than simply returning a uniformly random known distribution which would succeed with probability $1/k$. All proofs in this section are deferred to~\Cref{App:omitted-proofs-motivation}.

\begin{restatable}{lemma}{lightbad}
\label{lemma:weird}
For any $k\leq \binom{n}{n/2}$, there exist distinct distributions $p,q_1,\dots,q_{k}$ over $[n]$ such that $\|p-q_i\|_1=1$ for each $i\in [k]$ and if $\hat p$ is the empirical distribution obtained from sampling $s=n/2$ elements from $p$, then the probability that $\|p-\hat p\|_1\leq \min_{i\in [k]} (\|q_i-\hat p\|_1)$ is $\exp(-\Omega(k))$. In fact, if $k\geq Cn\log n$ for a sufficiently large constant $C$, the probability is $0$.
\end{restatable}

Another idea is to instead use an $\ell_2$ nearest neighbor data structure on the empirical distribution. Perhaps less surprisingly, this approach fails too, but whereas for $\ell_1$ the issue was the \emph{light elements}, for $\ell_2$ it is the \emph{heavy elements}. For simplicity, we consider the case of two distributions. 
\begin{restatable}{lemma}{heavybad}
\label{lemma:l2-not-working}
There exists distributions $p,q$ over $[n]$ with $\|p-q\|_1=1$ such that for any sample size $s=o(n)$, with probability $\Omega(1)$, $\|p-\hat p\|_2>  \|q-\hat p\|_2$.
\end{restatable}
Finally, it should be clear that $\ell_\infty$ nearest neighbor search on the empirical distribution $\hat p$ does not work with $s=o(n)$ samples. Indeed, consider distributions where all probabilities are either $\frac{2}{n}$ or $0$ and the most sampled coordinate $i$. Even with an exact nearest neighbor data structure, we could return an arbitrary distribution $q$ such that $q(i)=2/n$.

\section{Main algorithm and proof intuition}\label{sec:main_algo}
As motivated by the previous section, our algorithm deals with `heavy' and `light' elements separately. The first challenge is to formally define the notions of heavy and light elements. The most natural choice is to declare a domain element heavy for a distribution $v$ if the probability mass $v$ places on the element is larger than some threshold $\gamma$. An issue arises when we want to employ this definition for $p$, the distribution we are sampling from. Since we only get sample access to $p$, we will have tiny, but non-negligible, error on estimates of the probabilities of $p$, which introduces some ambiguity on the exactly partition of the domain into heavy and light elements. Thus, we are motivated seek an unequivocal definition of heavy and light. 

Towards this goal, we can define a heavy and light partitioning of $[n]$ according to one of the fixed distributions $v_i$, which we know fully, via Algorithm \ref{alg:heavy_light_decomposition}. It partitions the domain $[n]$ with respect to a fixed threshold $\gamma$ based on the probability masses of $v_i$. When we sample from $p$, we wish to use one of these partitions of the domains (obtained from inputting the $v_i$'s into Algorithm \ref{alg:heavy_light_decomposition}) as the definition of heavy and light elements for $p$. However, we still wish to ensure that any element that we define as heavy for $p$ does not possess probability mass significantly greater than $\gamma$. Lemma \ref{lem:p_p_hat_l_infinity} proves that $\hat{p}$ and $p$ are $\tilde{O}(1/\sqrt{n})$ close in $\ell_{\infty}$, where $\hat{p}$ is the empirical distribution after sampling from $p$.
Let $v^{\infty}$ denote the closest distribution to $\hat{p}$ among the $\{v_i\}_{i=1}^k$ in $\ell_{\infty}$ distance.
As there exists a choice of $v^\infty$ (namely $v^*$) which is $O(1/\sqrt{n})$ close to $p$ (and $\hat p$) in $\ell_\infty$, $\|v^\infty - p\|_\infty \leq O(1/\sqrt{n})$.
Hence, if we use the heavy/light partitioning of $v^{\infty}$, then every domain element classified as heavy for $p$ has probability mass at least $\gamma - O(1/\sqrt{n})$ and no light element has mass more than  $\gamma + O(1/\sqrt{n})$.

The above discussion naturally leads us to our preprocessing algorithm, Algorithm \ref{alg:preprocessing}. In Algorithm \ref{alg:preprocessing}, we first group all the known distributions $\{v_i\}_{i=1}^{k}$ into $k$ groups such that the $i$th group consists of all distributions which are close to $v_i$ in $\ell_{\infty}$ distance. Let $v^{\infty}$ denote the closest distribution to $\hat{p}$ in $\ell_{\infty}$ and consider $v^{\infty}$'s group, denoted as $S_{\infty}$. We essentially use the heavy / light partitioning of the domain induced by $v^{\infty}$ as the definition of heavy or light for $p$. All distributions in $v^{\infty}$'s group will also share the same heavy / light partitioning as $v^{\infty}$. By adjusting the log factors, we can ensure that $v^*$ is also included in $v^{\infty}$'s group. Hence, we only need to consider the distributions in $S_{\infty}$.

The key observation is that we can now \emph{completely} disregard the heavy elements. This is because $p$ will be sufficiently close to \emph{all} distributions in $S_{\infty}$ in $\ell_1$ distance, restricted to the heavy elements. This is formalized in Lemma \ref{lem:p_close_l_1_heavy}.

Thus, it suffices to consider the distributions in $S_{\infty}$, restricted to the light domain elements. It turns out that $\ell_2$ is a suitable metric to use on the light elements. Therefore, our main algorithm simply performs an $\ell_2$ nearest neighbor using query $\hat{p}$ and searches for the closest distribution in $S_{\infty}$ to $\hat{p}$, restricted to the light domain elements. Since we know $v^*$ is in $S_{\infty}$, we can guarantee that we will find a distribution which is `close' in $\ell_1$ distance on the light elements. We conclude by combining the fact that all distributions in $S_{\infty}$ are close to $p$ on the heavy elements already, and thus the total $\ell_1$ error must be small.

To summarize, the overall algorithm outline is the following. In the preprocessing step, we create an $\ell_{\infty}$ data structure on all the known distributions $\{v_i\}_{i=1}^k$. Then we group all of the known distributions into $k$ groups, one for each $v_i$, with the property that all distributions in a fixed group are $\tilde{O}(1/\sqrt{n})$ close in $\ell_{\infty}$ distance. We additionally instantiate an $\ell_2$ nearest neighbor data structure within each group, but only on the light domain elements, which is defined consistently within each group (using $v_i$ for $v_i$'s group). See Algorithm \ref{alg:preprocessing} for details. 

For the final algorithm, presented formally in Algorithm \ref{alg:main_algo}, we take $s$ samples from the unknown distribution $p$ and let $\hat{p}$ denote the empirical distribution. We query the $\ell_{\infty}$ nearest neighbor search data structure on $\hat{p}$ which helps us narrow down to a fixed group of distributions close on the heavy elements, as defined in the preprocessing stage. We then query the $\ell_2$ NNS data structure for this group on $\hat{p}$, but restricted to the light domain elements in the group\footnote{To ensure independence, we actually query using an empirical distribution using fresh samples.}. The returned distribution is our output for the closest distribution to $p$. The main guarantee of our algorithms is given below.
\begin{restatable}{theorem}{mainthm}
\label{thm:main}
Set $s = \Theta\left(\frac{n}{\eps^2(\log k)^{1/4}}\right) = o(n)$ in Algorithm \ref{alg:main_algo} and $\gamma = 1/n^{5/12}$ in Algorithm \ref{alg:preprocessing}. Let $\tilde{v}$ denote the output of Algorithm \ref{alg:main_algo}. Then the preprocessing algorithm, Algorithm \ref{alg:preprocessing}, runs in time polynomial in $k$, requires $\tilde{O}(nk^2)$ space, and the query time of Algorithm \ref{alg:main_algo} is $\tilde{O}(n) + k^{1 - 1/(\log k)^{1/4}} = o(k)$. Furthermore, with probability $1-o(1)$, Algorithm \ref{alg:main_algo} returns distribution $\tilde{v}$ satisfying $\|p - \tilde{v}\|_1 \le \eps$.
\end{restatable}

In summary, we use polynomial preprocessing time, $s = o(n)$ samples, and our query time is $o(k)$, with the latter two quantities being sublinear in the domain size and the number of distributions, respectively. We present auxiliary lemmas in Sections \ref{sec:heavy_proofs} and \ref{sec:light_proofs} (proofs in~\Cref{app:heavy-lemmas,app:light-lemmas}, respectively) and prove Theorem \ref{thm:main} in~\Cref{sec:main_thm_proof}.

\begin{algorithm}[!htb]
\caption{\label{alg:heavy_light_decomposition} Heavy light decomposition }
\begin{algorithmic}[1]
\State \textbf{Input}: An ordered list of distributions $q_1, \ldots, q_t$
\State \textbf{Output}: Vectors $(q_i)_H, (q_i)_L$ for every $i \in [t]$ (recall Definition \ref{def:vector_A})
\Procedure{Heavy-Light}{$\{q_i\}_{i=1}^t$}
\State $H, L \gets \emptyset$
\For{$j = 1$ to $n$}
\If{$q_1(j) \ge \gamma$}
\Comment{We use $q_1$ to define the heavy and light elements for all $q_1, \ldots, q_t$.}
\State{$H \gets H \cup \{j\}$}
\Else{}
\State{$L \gets L \cup \{j\}$}
\EndIf
\Comment{$H \cup L$ is a disjoint partition of $[n]$.}
\EndFor
\State \textbf{Return:} vectors $\{(q_i)_H, (q_i)_L\}_{i=1}^t$ 
\Comment{We use the distribution $q_1$ to set the heavy and light domain elements for all other distributions.}
\EndProcedure
\end{algorithmic}
\end{algorithm}

\begin{algorithm}[!htb]
\caption{\label{alg:preprocessing}Preprocessing}
\begin{algorithmic}[1]
\State \textbf{Input}: Distributions $\{v_i\}_{i=1}^k$ according to Section \ref{sec:prelim}
\State \textbf{Output}: An $\ell_{\infty}$ data structure, a number of $\ell_{2}$ NNS data structures such that each $v_i$ is associated with a unique $\ell_2$ data structure
\Procedure{Preprocessing}{$\{v_i\}_{i=1}^k$}
\State $D^{\infty} \gets \ell_{\infty}$ data structure  on $\{v_i\}_{i=1}^k$ 
\For{all $i \in [k]$}
\State $S_i \gets \{v' \in \{v_i\}_{i=1}^k \mid \|v_i - v'\|_{\infty} \le O((\log n)^2 \cdot (\log \log n)/ \sqrt{n}) \}$
\EndFor
\For{$j = 1$ to $k$}
\State Write $S_j = \{w_1, \ldots, w_t\}$ where $w_1 = v_j$
\State $\{(w_i)_H, (w_i)_L\}_{i=1}^t \gets$ \texttt{Heavy-Light}$(S_j)$ 
\State $D_j \gets \ell_{2}$ NNS data structure on $\{ (w_i)_L\}_{i=1}^t$
\Comment{$S_j$ corresponds to a well defined partition of $[n]$ into heavy and light elements based on $v_j$ via Algorithm \ref{alg:heavy_light_decomposition}}
\EndFor
\State \textbf{Return:} $D^{\infty}$, data structures $\{D_j\}_{j=1}^k$
\EndProcedure
\end{algorithmic}
\end{algorithm}

\begin{algorithm}[!htb]
\caption{\label{alg:main_algo} Sublinear Time Hypothesis Selection}
\begin{algorithmic}[1]
\State \textbf{Input}: Distributions $\{v_i\}_{i=1}^k$; preprocessed data structures $D^{\infty}, \{D_j\} \gets \texttt{Preprocessing}(\{v_i\}_{i=1}^k)$; $\text{Poi}(s)$ samples from query distribution $p$
\State \textbf{Output}: A distribution $v_j$
\Procedure{Sublinear-Hypothesis-\\Selection}{$\{v_i\}_{i=1}^k$}
\State $\hat{p} \gets$ empirical distribution of the first half of the samples from $p$
\State $v^{\infty} \gets$ output of $D^{\infty}$ on query $\hat{p}$ with approximation $c = O(\log(n) \cdot \log \log n) $
\State $D' \gets$ $\ell_2$ NNS data structure corresponding to $v^{\infty}$
\State $L \gets$ the light domain elements for $D'$ 
\State $\hat{p}' \gets$ empirical distribution of the second half of the samples from $p$
\State $\tilde{v} \gets$ output of $D'$ on $\hat{p}'_L$ with approximation $c = 1 + \frac{s \eps^2}{32n}$
\State \textbf{Return:} $\tilde{v}$
\EndProcedure
\end{algorithmic}
\end{algorithm}

\subsection{Auxiliary lemma for heavy elements}\label{sec:heavy_proofs}
Lemma \ref{lem:num_heavy_elems} shows that any pair of distributions in a fixed group, as defined in the preprocessing step, are close in $\ell_1$ distance, when the $\ell_1$ distance is restricted to the heavy domain elements in the group.

\begin{restatable}{lemma}{numheavyelemes}\label{lem:num_heavy_elems}
    Suppose $\gamma = 1/n^C$ in Algorithm \ref{alg:heavy_light_decomposition}. Consider the sets $S_j = \{w_1, \ldots, w_t\}$ defined in line 9 of Algorithm \ref{alg:preprocessing}. Consider the vectors $(w_i)_{H}$, which are the heavy subsets of the distributions in $S_j$, as defined in line 10 of Algorithm \ref{alg:preprocessing}. Then for all $j$ and for all $w,w' \in S_j$, we have \newline
$\|w_H - w'_H\|_1 \le O\left(\frac{(\log n) \log \log n}{n^{1/2 - C}} \right).$
\end{restatable}

Lemma \ref{lem:opt_in_group} shows that $v^*$, the distribution close to $p$ as defined in Section \ref{sec:prelim}, must belong to the same group as $v^{\infty}$, the distribution returned after querying $\hat{p}$ in the $\ell_{\infty}$ data structure in Algorithm \ref{alg:main_algo}.

\begin{restatable}{lemma}{optingroup}\label{lem:opt_in_group}
    Suppose $s \ge \Omega(n / (\log n)^{1/2})$. Consider the output $v^{\infty}$ on $\hat{p}$ when inputted into the $\ell_{\infty}$ data structure $D^{\infty}$, as done in line $6$ of Algorithm \ref{alg:main_algo}. Let $S_{\infty}$ denote the group of $v^{\infty}$ from Algorithm \ref{alg:preprocessing}. Assuming the event in Lemma \ref{lem:p_p_hat_l_infinity} holds, it must be that $v^* \in S_{\infty}$.
\end{restatable}

The final~\ref{lem:p_close_l_1_heavy} below argues that $S_{\infty}$, the group of $v^{\infty}$ in the preprocessing algorithm, has the property that $p$ must be close in $\ell_{\infty}$ distance to every member of $S_{\infty}$. Consequently, $p$ must be close to every distribution in $S_{\infty}$ in $\ell_1$ distance, restricted to the heavy elements of the group.

\begin{restatable}{lemma}{pcloseloneheavy}\label{lem:p_close_l_1_heavy}
    Consider the same setting as in Lemma \ref{lem:opt_in_group} and suppose $\gamma = 1/n^C$. Let $H$ denote the heavy domain elements of group $S_{\infty}$. Then for all $w \in S_{\infty}$, we have
$ \|p - w\|_{\infty} \le O\left( \frac{(\log n)^2  \log \log n}{\sqrt{n}} \right). $
Consequently, we also have
$ \|p_H - w_H\|_{1} \le O\left( \frac{(\log n)^2  \log \log n}{n^{1/2 - C}} \right).$
\end{restatable}
Essentially, the lemmas above ensure that $v^*$ will be in the group $S_\infty$ which in turn has the property that for all the distributions are close in $\ell_1$ on the heavy elements. In order to actually find the closest distribution in $\ell_1$, we have to switch our attention to the light elements within the group and for this we need the lemmas of the next section.

\subsection{Auxiliary lemmas for light elements}\label{sec:light_proofs}
We now state auxiliary lemmas concerning light elements. Recall the definition of $L$ from line $8$ in Algorithm \ref{alg:main_algo}. We first show the expected value of $\|s \cdot \hat{p}'_L-s \cdot v_L\|_2$ captures the $\ell_2$ distance between $p$ and $v$, restricted to the light elements. Note that we are using the empirical distribution $\hat{p}'$ which does not share any samples with $\hat{p}$, insuring the independence of $L$ and $\hat{p}'$.  We also define $T$ in the lemma below as $T = \sum_{i \in L } p(i)$.

\begin{restatable}{lemma}{lightEV}
\label{lem:light_EV}
$\mathbb{E}[\|s \cdot \hat{p}'_L-s \cdot v_L\|_2^2] = s \cdot T + s^2 \|p_L-v_L\|_2^2.$
\end{restatable}

We now derive concentration of the estimator around its expected value. First we consider the case where $v$ is sufficiently close to $p$.

\begin{restatable}{lemma}{lightconcentrationone}\label{lem:light_concentration_1}
   Suppose $p$ satisfies $\|p_L-v_L\|_2  \le  \frac{\eps}{2\sqrt{n}}$. Set $t = \frac{s^2 \|p_L-v_L\|_2^2}{4} + \frac{s^2 \eps^2}{4n}$
and $Z_1 = \|s \cdot \hat{p}'_L-s \cdot v_L \|_2^2$. If $s = \Omega\left(\max\left(\frac{n \|p_L\|_2}{\eps^2}, \frac{n^{2/3}}{\eps^{4/3}}\right) \right)$
then $\mathbb{P}[| Z_1 - \mathbb{E}[Z_1] | \ge t] \le 0.01$. 
\end{restatable}
Next we consider the case where $\eps/\sqrt{n} \lesssim \|p_L-v_L\|_2$. In this case, we need to obtain a stronger concentration result since later we union bound over possibly $\Omega(k)$ different distributions which are in group $S_{\infty}$.

\begin{restatable}{lemma}{lightconcentrationtwo}\label{lem:light_concentration_2}
    Suppose $p$ satisfies $\|p_L-v_L\|_2  \ge  0.75 \cdot \frac{\eps}{\sqrt{n}}$. Set 
$t = \frac{s^2 \|p_L-v_L\|_2^2}{4} + \frac{s^2 \eps^2}{4n},$
suppose $\gamma = 1/n^{5/12}$,
and let $Z_1 = \|s \cdot \hat{p}'_L-s \cdot v_L \|_2^2$ denote our estimator.
If $s = \Omega\left(\ \frac{n^2\gamma^{3}\log(k)^2}{\eps^4} \right)$
then $\mathbb{P}[| Z_1 - \mathbb{E}[Z_1] | \ge t] \le 1/\textup{poly}(k)$.
\end{restatable}

\paragraph{Putting it all together: Proof of Theorem \ref{thm:main}}

We give a high level outline of the proof by showing how to combine the prior results to prove the main theorem and defer the full proof to Appendix \ref{sec:main_thm_proof}. See Section \ref{sec:main_algo} for additional inituition.

Lemma \ref{lem:p_close_l_1_heavy} states that we can essentially ignore the heavy elements in the group $S_{\infty}$ and furthermore, Lemma \ref{lem:opt_in_group} proves that $v^* \in S_{\infty}$. Thus we can restrict our attention to the light elements of the group belonging to $v^{\infty}$. Lemma \ref{lem:light_EV} states that the expected value of the (scaled) $\ell_2$ distance squared from $\hat{p}'_L$, the empirical distribution, to $v^*_L$ is at most $s + s^2 \|p_L - v^*_L\|_2^2 \le s + s^2\eps^2/(4n)$ (assuming $T$ is a constant for the sake of simplicity of the discussion). On the other hand, the expected value of the same statistic for any distribution $v$ in $S^{\infty}$ which is $0.99 \eps$ far from $p$ in $\ell_1$ has to be at least $s + s^2\|p_L - v_L\|_2^2 \ge s + (0.99)^2s^2 \eps^2 /n$. This is bigger than the corresponding value for $v^*$ by a factor of roughly $1+O(\eps^2 s^2/n)$. These calculations are only true in expected value so we use Lemmas \ref{lem:light_concentration_1} and \ref{lem:light_concentration_2} to prove that the quantities are close to their expected value with high probability. Since the values of the $\ell_2$ statistic are sufficiently far apart in these two cases, the $\ell_2$ NNS data structure outputs a close distribution $\tilde{v}$ satisfying the theorem guarantees.

\section{Improving the runtime for classic hypothesis testing}
\label{sec:fast-tournament}
We now switch our focus to algorithms for the classic hypothesis selection problem in the \emph{improper} setting. Recall from Section \ref{sec:prelim} that we are given a set of $k$ distributions $\mathcal{V}=\{v_1,\dots,v_k\}$ over $[n]$ and a set of samples $\mathcal{S}$ from some distribution $p$ also over $[n]$. Given the samples, our goal is to output $\hat v\in \mathcal{V}$ such that 
$
\|p-\hat v\|_1 \le C \cdot \min_j \|p- v_j\|_1 + \eps $
with probability at least $1-\delta$ for some constant $C$ and some small 
parameters $\eps$ and $\delta$. 
In the case $k=2$, this problem can be solved using the so called Scheffe test~\cite{Scheffe1947useful} which we will discuss shortly. This test uses $O(\frac{\log 1/\delta}{\eps^2})$ samples and returns a $\hat v\in \mathcal{V}$ satisfying the above bound with $C=3$.
In~\cite{devroye2001combinatorial} this was extended to the case $k>2$. They showed that with $s=O((\log k+\log \frac{1}{\delta})/\eps^2)$ samples, running the Scheffe test for each pair of distributions and outputting the one with the most wins, yields an estimate that satisfies the bound above with $C=9$. This is referred to as the Scheffe tournament. The running time for this tournament, clearly scales with $O(k^2)$. However, using a knockout tournament style algorithm,~\cite{suresh2014spherical} showed that the running time can be reduced to $O\left(\frac{k}{\eps^2}(\log k+\log \frac{1}{\delta})\right)$ for a fairly large constant $C$ even when we are only given sample access to the distributions in $\mathcal{V}$. In~\cite{acharya2018maximum} the authors showed that a simple algorithm \texttt{Quick-Select} obtains $C=9$ with the same running time and sample size. Here we show how to reduce this runtime to $O\left(\frac{k}{\eps^2}\log \frac{1}{\delta}\right)$ with a slight blow-up in the value of $C$. For instance, when $\delta=\Omega(1)$ this shaves off a factor of $\log k$ of the running time. To obtain this bound, we will allow our algorithm a simple preprocessing step, namely for each pair of distributions $(v_i,v_j)$, we compute and store the total variation distance between them. This can be done in total time $O(k^2 n)$. 

To achieve this improvement, we recall the classic Scheffe test~\cite{Scheffe1947useful}. This test can be seen as an algorithm for the hypothesis selction problem for the case $k=2$. The algorithm takes as input two distributions $v_1$ and $v_2$ over $[n]$ and samples from an unknown distribution $p$ over $[n]$. Let's define $S=\{i\in [n]:v_1(i)> v_2(i)\}$. Upon receiving the samples, the algorithm computes $\mu_S$ the frequency of samples in $S$. It then outputs $v_1$ if $|v_1(S)-\mu_S|\leq |v_2(S)-\mu_S|$. Otherwise, it outputs $v_2$. Note that the values $v_1(S)$ and $v_2(S)$ can be computed directly from the total variation distance between $v_1$ and $v_2$. It was shown in~\cite{devroye2001combinatorial} that for $k=2$ and with $s$ samples, with probability $1-\delta$ the Scheffe test outputs a distribution $v$ satisfying that
\begin{align}\label{eq:scheffe-estimate}
\|p-v\|_1 \le 3 \cdot \min_{j\in \{1,2\}} \|p- v_j\|_1 + \sqrt{\frac{10\log \frac{1}{\delta}}{s}}.
\end{align}
Except for a simple modification, our algorithm is similar to the algorithm proposed in~\cite{suresh2014spherical}. Assume for simplicity that $k$ is a power of two --- this assumption can easily be dispensed with. Define $\delta_i=\delta/4^i$ and $s_i=\frac{10\log (1/\delta_i)}{\eps^2}$. Finally, define $\mathcal{V}_1=\mathcal{V}$. Our algorithm first initializes a set $\mathcal{C}\leftarrow\emptyset$. Then for $i=1,\dots, \lg k$, it performs the following steps:
\begin{enumerate}
  \setlength\itemsep{0em}
    \item Randomly select a subset of $\min\{k^{1/3},|\mathcal{V}_i|\}$ elements from $|\mathcal{V}_i|$ and move them to $\mathcal{C}$.
    \item Randomly form $|\mathcal{V}_i|/2$ pairs of distributions in $\mathcal{V}_i$ and run the Scheffe test on each pair using the set $\mathcal{S}_i$ consisting of the first $s_i$ elements of the sample. 
    \item Define $\mathcal{V}_{i+1}$ to be the set of $|\mathcal{V}_i|/2$ winners.
\end{enumerate}
Finally, our algorithm runs the Scheffe test on all pairs of distributions in $\mathcal{C}$ using a single set of $\frac{10\log \left(\binom{|\mathcal{C}|}{2}/\delta\right)}{\eps^2}$ new samples from $p$ (independent of the samples used in step 2. above). It outputs the distribution $\hat v$ in $\mathcal{C}$ with the most wins among these comparisons (breaking ties arbitrarily). 

The difference between this algorithm and the algorithm in~\cite{suresh2014spherical} is that we do not use the entire sample in the lower levels of the tournament tree. Intuitively, for a subtree of the tournament tree containing $2^i$ distributions, we only need to consider a sample of size $s_i$ in order to union bound over the bad events that the distribution in $\mathcal{V}$ closest to $p$ loses to either of the `far' distributions in this subtree. As running a single Scheffe test with a sample of size $s$ takes $O(s)$ time, the running time used for the knockout tournament is therefore $O\left(\sum_{i=1}^{\lg k} s_ik/2^i\right)=O\left(\frac{k}{\eps^2}\log \frac{1}{\delta}\right)$. We defer the full proof to Appendix \ref{sec:fast_tournament_proof}.

\begin{restatable}{theorem}{fasttournament}\label{thm:fast_tournament}
    Assume that $\delta\geq k^{-1/4}$. With probability $1-O(\delta)$, the algorithm of Section \ref{sec:fast-tournament} outputs a distribution $\hat v$ with
$
\|p-\hat v\|_1 \le 27 \cdot \min_j \|p- v_j\|_1 + O(\eps). 
$
The algorithm uses $s=O((\log k+\log \frac{1}{\delta})/\eps^2)$ samples from $p$ and has running time $O\left(\frac{k}{\eps^2}\log \frac{1}{\delta}\right)$.
\end{restatable}

\section{Experiments}
We experimentally evaluate the faster tournament algorithm given in \cref{sec:fast-tournament} and compare its performance to the base knockout tournament~\cite{suresh2014spherical} on synthetic and on real networking data.\footnote{Code available at \url{https://github.com/justc2/datastructdensityest}} These experiments display the practical benefits of the faster algorithm.

The algorithms have several key parameters which we vary throughout the experiments.
$\allpairs$ is the number of distributions in each level of the tournament which are randomly sampled to compete in an all-pairs tournament at the end of the algorithm. This parameter is used in both the base tournament and our fast tournament.
$\fast$ controls the number of samples used in each level of the fast tournament. In particular, at the $i$th level, the fast tournament uses $\fast \cdot i$ samples for each Scheffe test.
While adjusting these parameters slightly does not change the constant factors in the analysis in \cref{sec:fast-tournament}, we find they make a large impact empirically and thus test the algorithms under various parameter settings.

We measure computation cost by the number of Scheffe operations performed by each tournament. One Scheffe operation is one comparison of a pair of distributions at a given sampled element (the basic computation performed during a Scheffe test).
Below, we describe the experimental setup for the datasets.

\begin{figure}[ht]
 \centering
 \includegraphics[width=0.75\columnwidth]{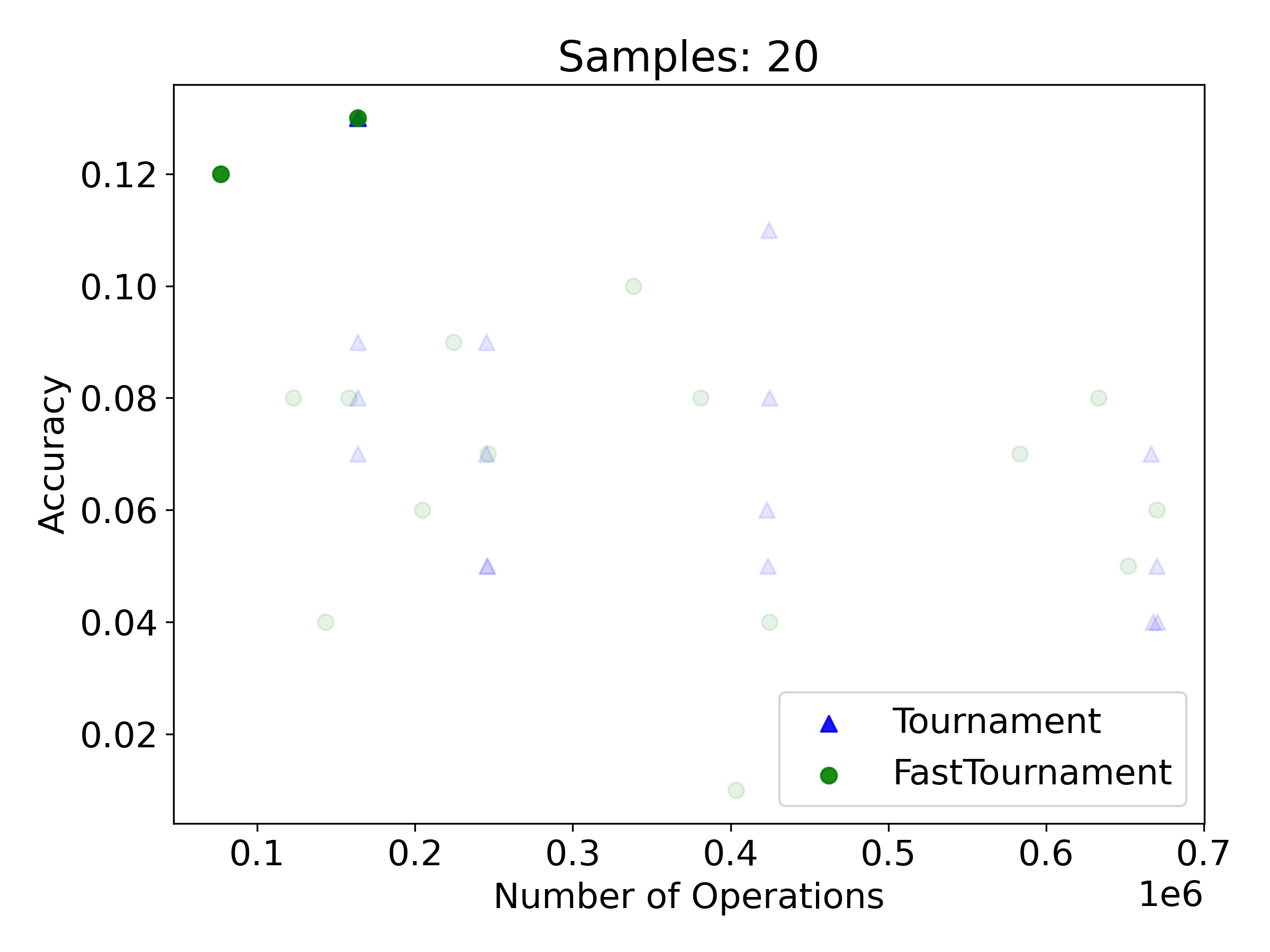}
 \caption{Grid search on half-uniform data with $20$ samples.}
 \label{fig:gs-20}
\end{figure}

\begin{figure}[ht]
 \centering
 \includegraphics[width=0.75\columnwidth]{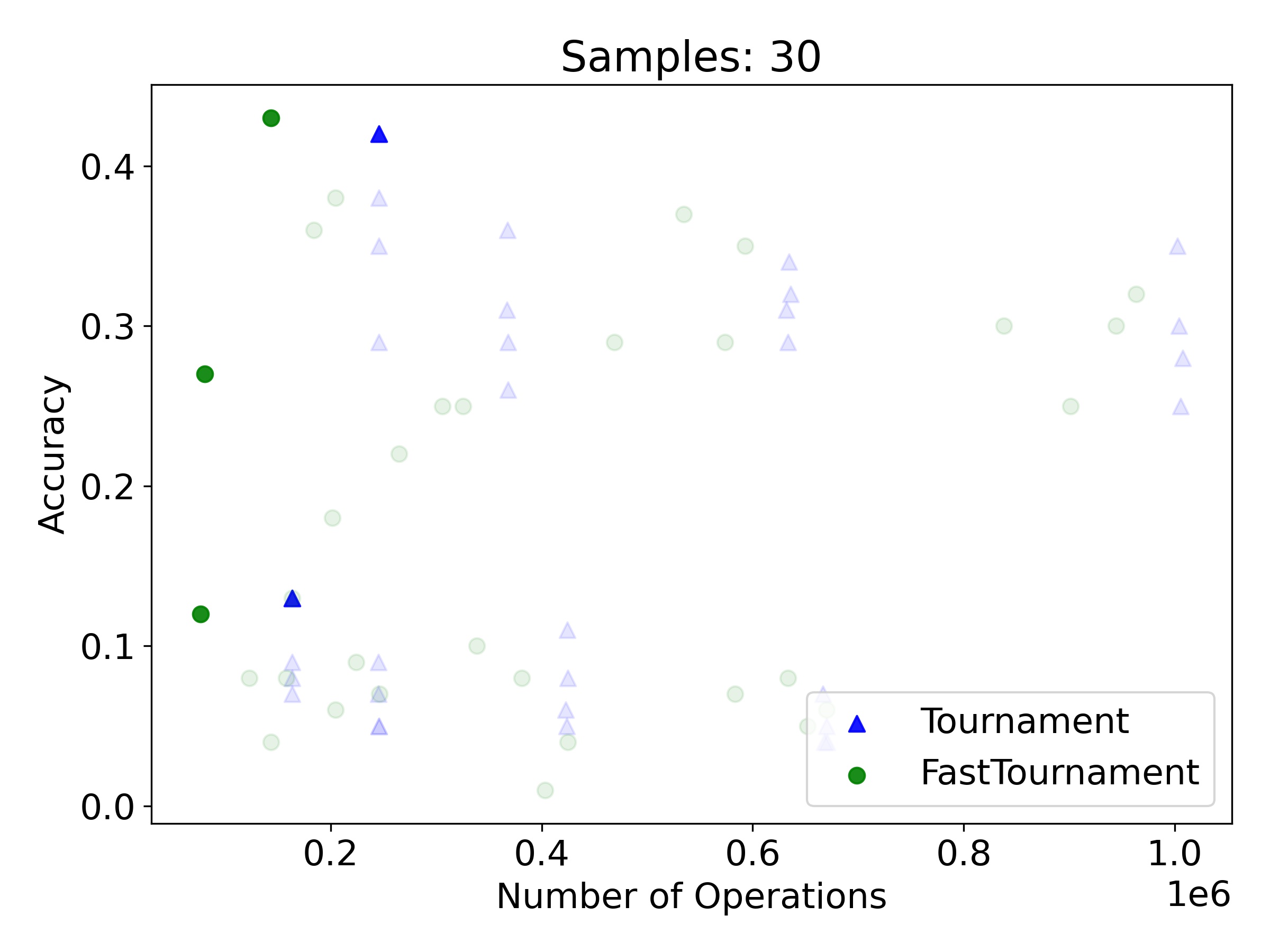}
 \caption{Grid search on half-uniform data with up to $30$ samples.}
 \label{fig:gs-30}
\end{figure}

\begin{figure}[ht]
 \centering
 \includegraphics[width=0.75\columnwidth]{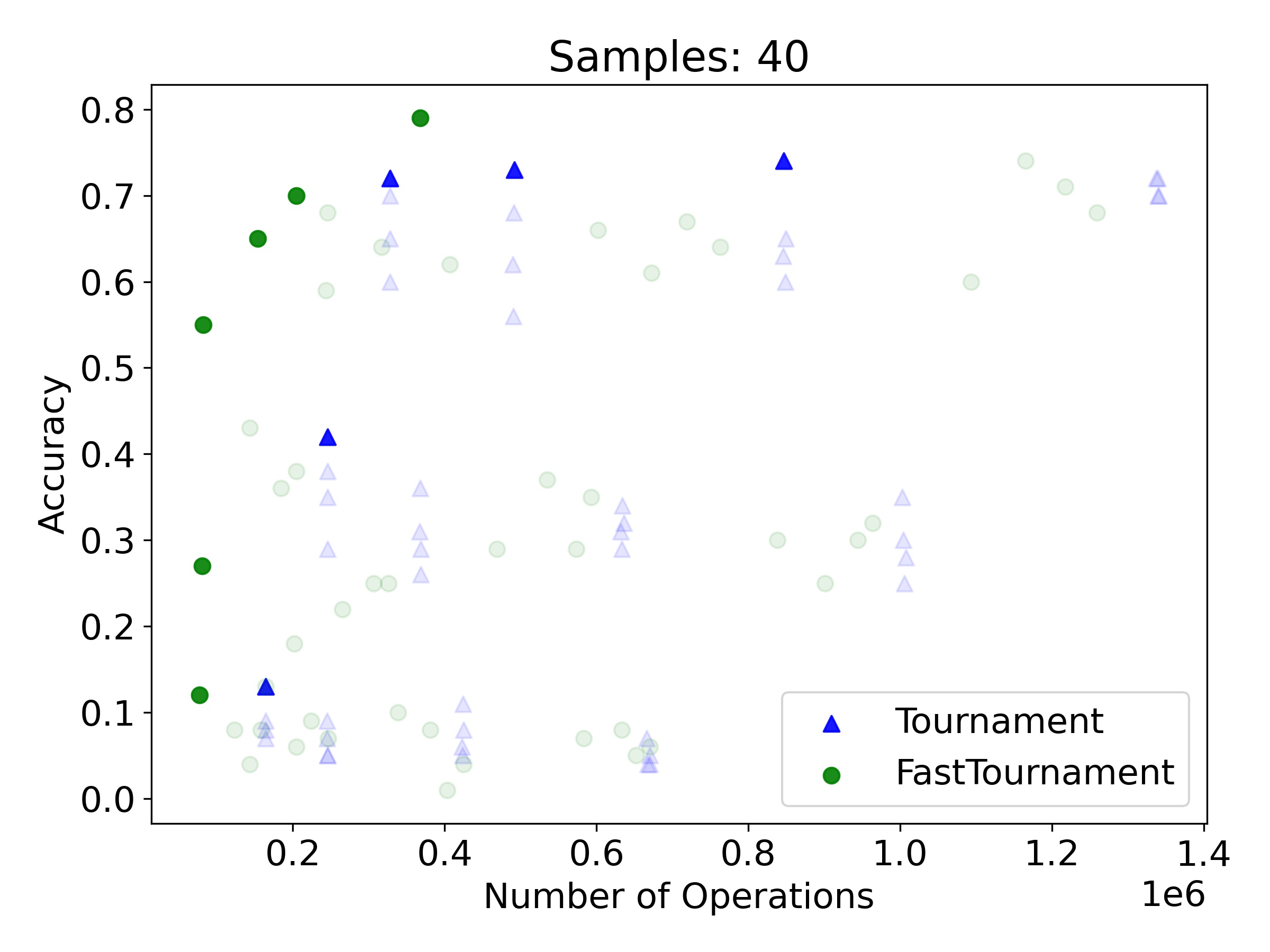}
 \caption{Grid search on half-uniform data with up to $40$ samples.}
 \label{fig:gs-40}
\end{figure}

\begin{figure}[ht]
 \centering
 \includegraphics[width=0.75\columnwidth]{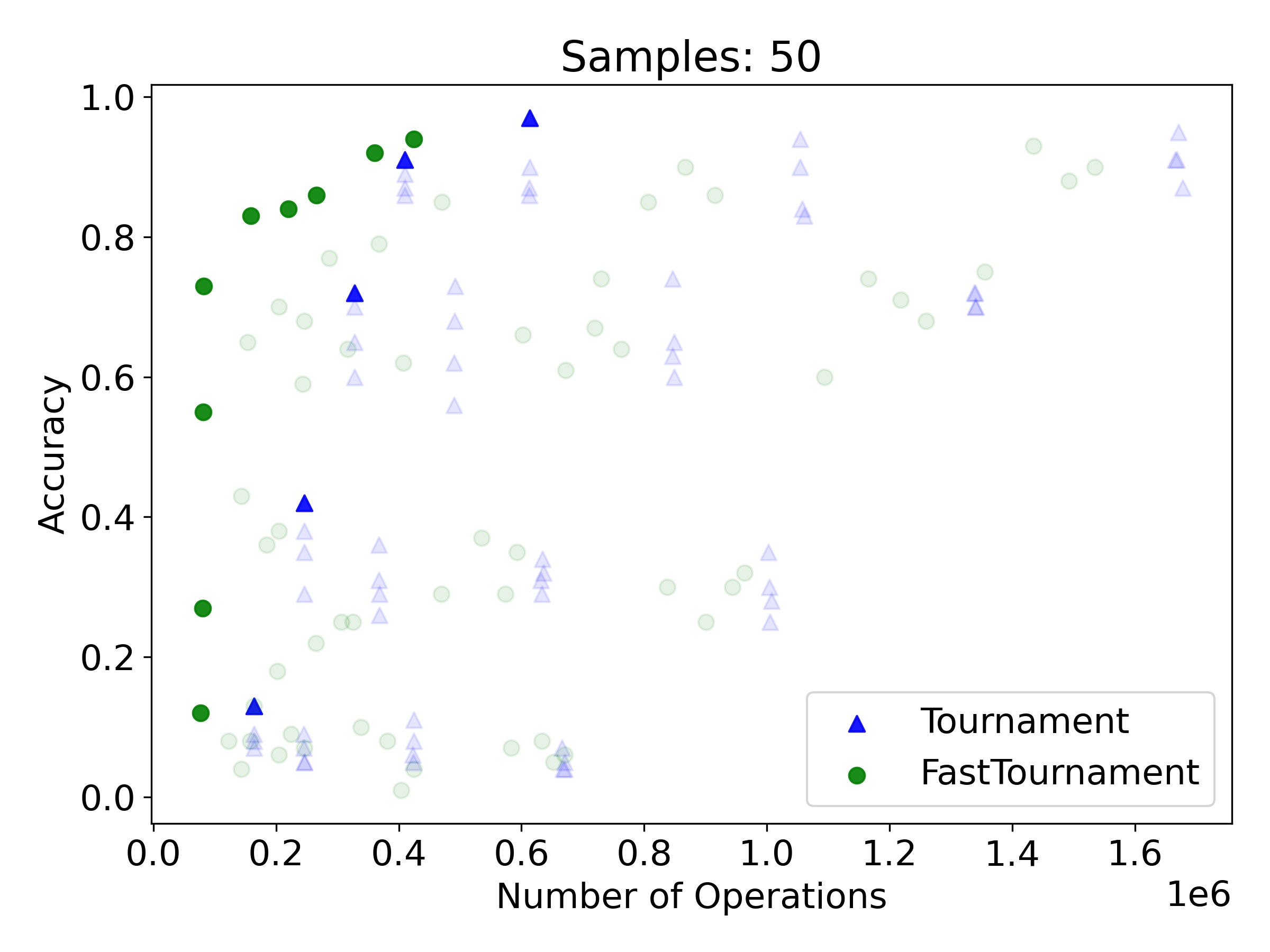}
 \caption{Grid search on half-uniform data with up to $50$ samples.}
 \label{fig:gs-50}
\end{figure}

\begin{figure}[ht]
 \centering
 \includegraphics[width=0.75\columnwidth]{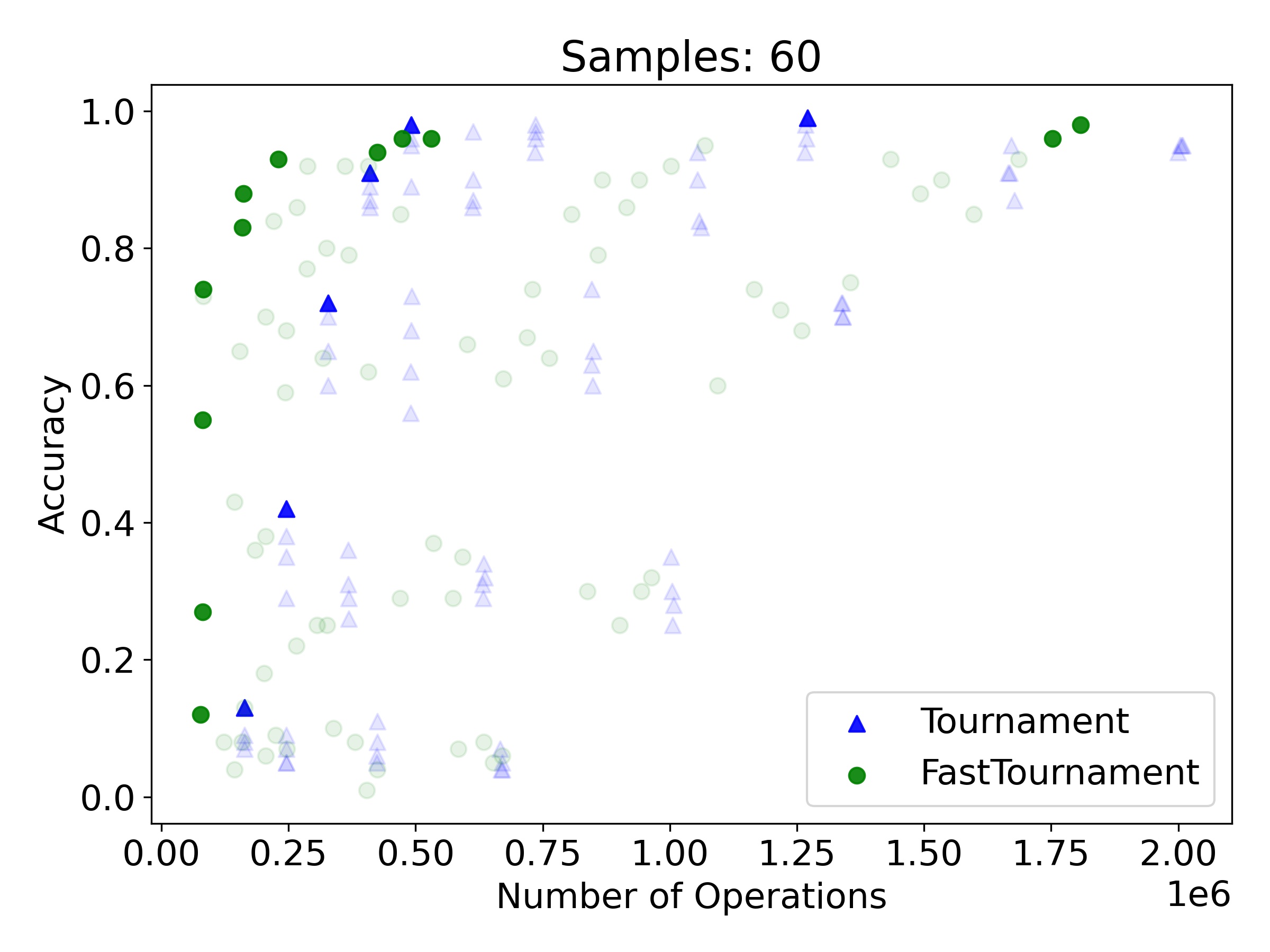}
 \caption{Grid search on half-uniform data with up to $60$ samples.}
 \label{fig:gs-60}
\end{figure}

\begin{figure}[ht]
 \centering
 \includegraphics[width=0.75\columnwidth]{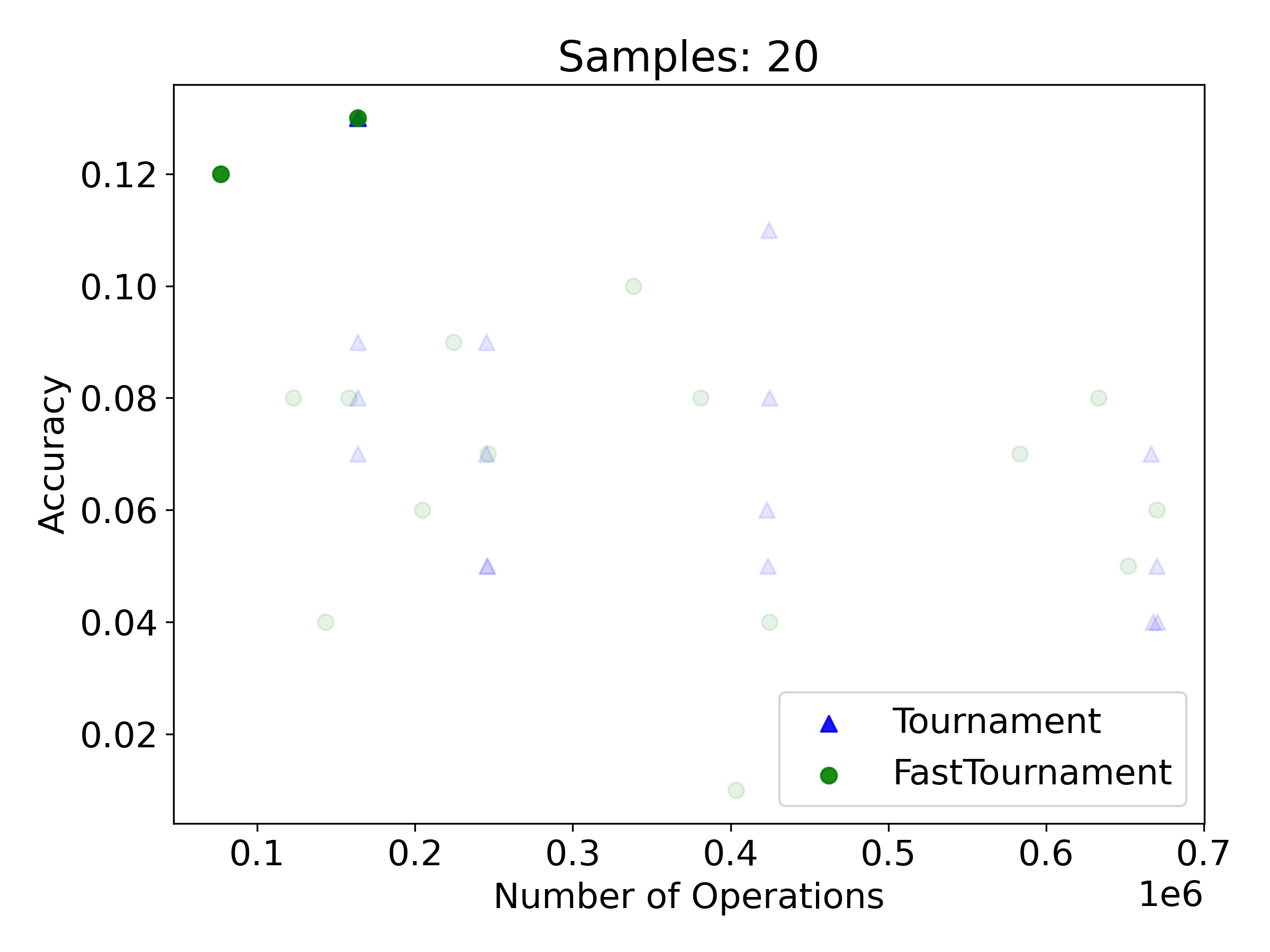}
 \caption{Grid search on Zipfian data with $20$ samples.}
 \label{fig:zipf-20}
\end{figure}

\begin{figure}[ht]
 \centering
 \includegraphics[width=0.75\columnwidth]{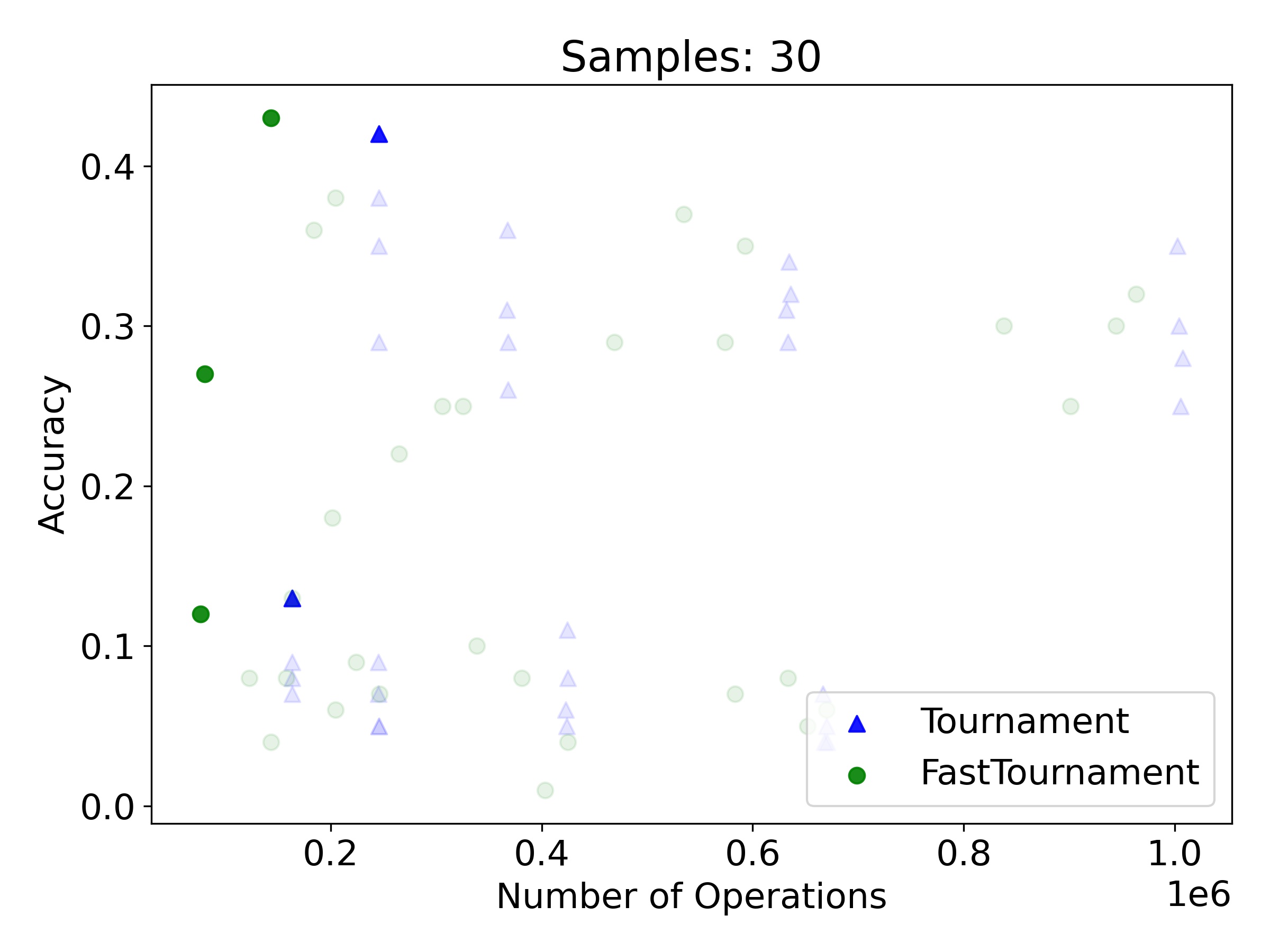}
 \caption{Grid search on Zipfian data with up to $30$ samples.}
 \label{fig:zipf-30}
\end{figure}

\begin{figure}[ht]
 \centering
 \includegraphics[width=0.75\columnwidth]{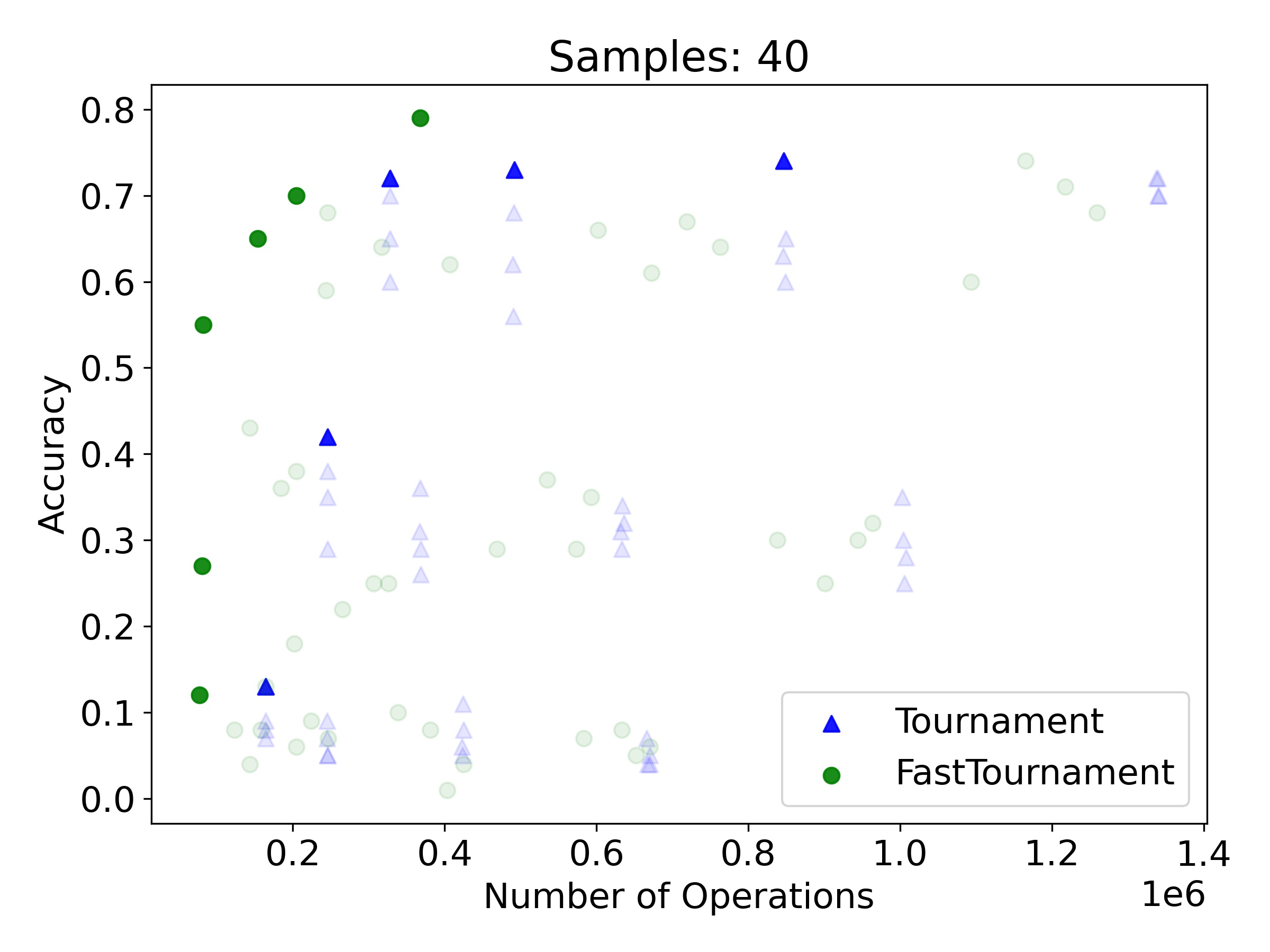}
 \caption{Grid search on Zipfian data with up to $40$ samples.}
 \label{fig:zipf-40}
\end{figure}

\subsection{Synthetic Experiments}
\paragraph{Setup}
We compare two synthetic datasets corresponding to \emph{half-uniform} and \emph{Zipfian} distributions.
The half-uniform dataset consists of $k=8192$ distributions over a domain of size $n=500$. Each distribution is uniform over a random $n/2$ sized subset of the domain. We consider a number of samples $s \in \{20, 30, 40, 50, 60\}$.

The Zipfian dataset consists of $k=4096$ distributions over a domain of size $n=250$. Each distribution is a random permutation of the standard Zipfian distribution where the $i$th element has probability proportional to $1/i$. We consider a number of samples $s \in \{20, 30, 40\}$.

In both cases, queries are formed by taking samples from one of the distributions in the dataset and performance is measured by the accuracy of the algorithms (i.e., the fraction of queried points for which the algorithm returned the true distribution).
As runtime can always be decreased by not using all of the sampled elements, for results using $s$ samples, we also report all results using fewer than $s$ samples.
For these experiments, we grid search over $\fast \in \{5, 10, 15, 20\}$ and $\allpairs \in \{0, 10, 20, 30\}$. 
The reported results are averaged over 5 random sets of 20 queries each.
\vspace{-3mm}
\paragraph{Results}
We will focus on the results for the half-uniform dataset as the conclusions for the Zipfian dataset are similar (see specific comments at the end of this subsection).
In Figures~\ref{fig:gs-20}-\ref{fig:gs-60}, we compare the accuracy and computational cost of the base tournament and our fast tournament under the various parameters setting described above. For both algorithms, the upper envelope of points (i.e., the discrete approximation of the Pareto curve for the accuracy/time tradeoff)  are highlighted.
As the number of samples increases, the accuracy both algorithms dramatically increases from $13\%$ at $20$ samples up to $98\%$ at $60$ samples.

Across different levels of sampling, the fast tournament is able to attain similar accuracy to the base tournament while using significantly fewer samples. 
For example, at $60$ samples, the fast tournament is able to achieve accuracy of $88\%$ using fewer than $165{,}000$ operations while the slow tournament requires more than $400{,}000$ operations to achieve accuracy greater than $80\%$ (an improvement of more than $2.4\times$).
In general, we find that the most important factor influencing accuracy is the number of total samples. On the other hand, using fewer samples in earlier rounds of the tournament via small $\fast$ has a moderate to negligible impact on accuracy. Finally, the points of the right side of the plots with many operations correspond to larger values of $\allpairs$. Increasing $\allpairs$ (at least beyond a certain point) has small impact on accuracy while significantly increasing computation due to the quadratic dependence on the all-pairs comparison at the end of the tournament.

The overall results for the Zipfian datset in Figures~\ref{fig:zipf-20}-\ref{fig:zipf-40} are very similar though the Zipfian distributions are qualitatively different from the half-uniform distributions as they contain several elements with quite large probabilities.

\begin{figure}[!htbp]
     \centering
     \begin{subfigure}[b]{0.75\columnwidth}
         \centering
         \includegraphics[width=\textwidth]{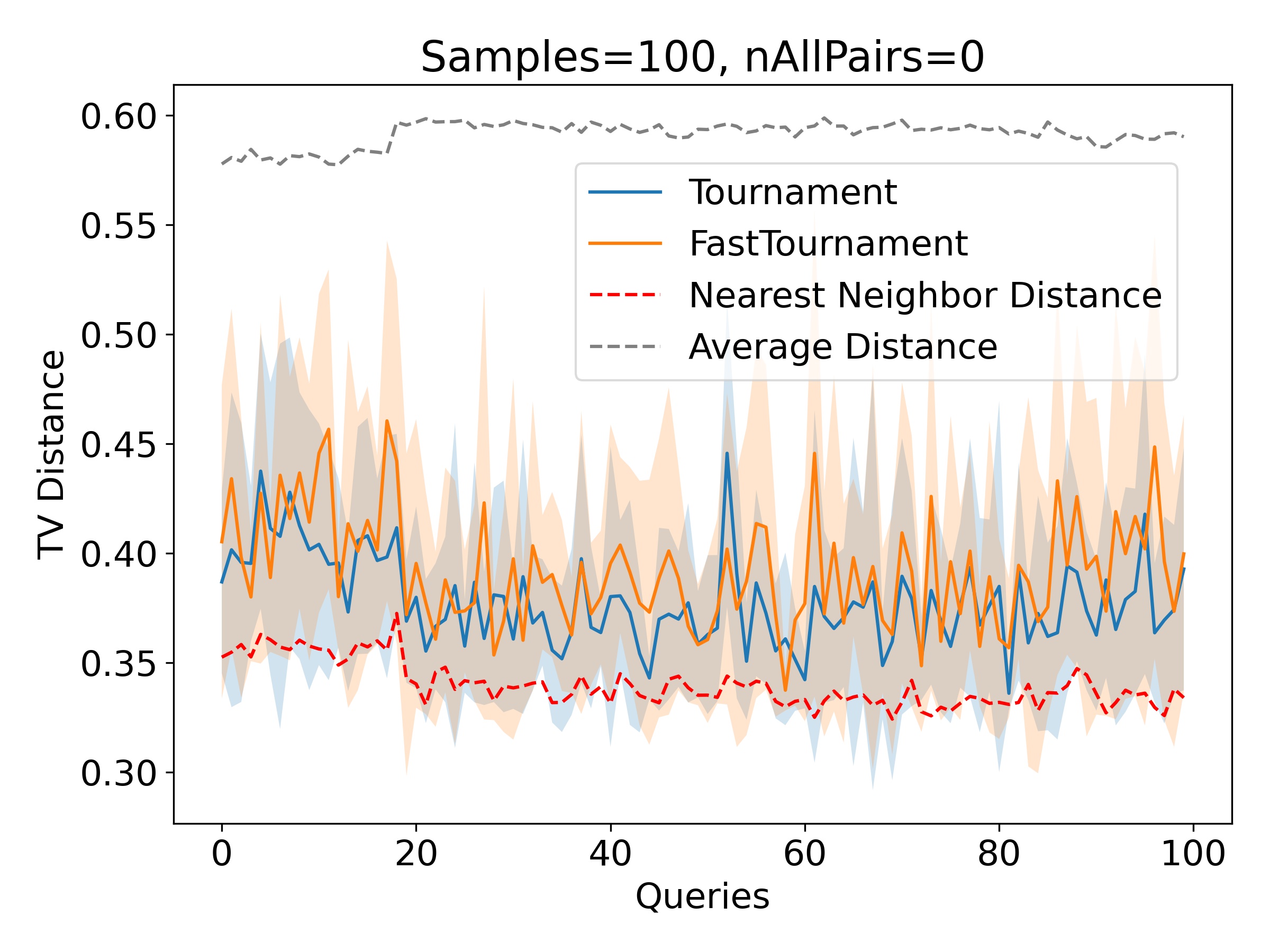}
     \end{subfigure}

    \centering
     \begin{subfigure}[b]{0.75\columnwidth}
         \centering
         \includegraphics[width=\textwidth]{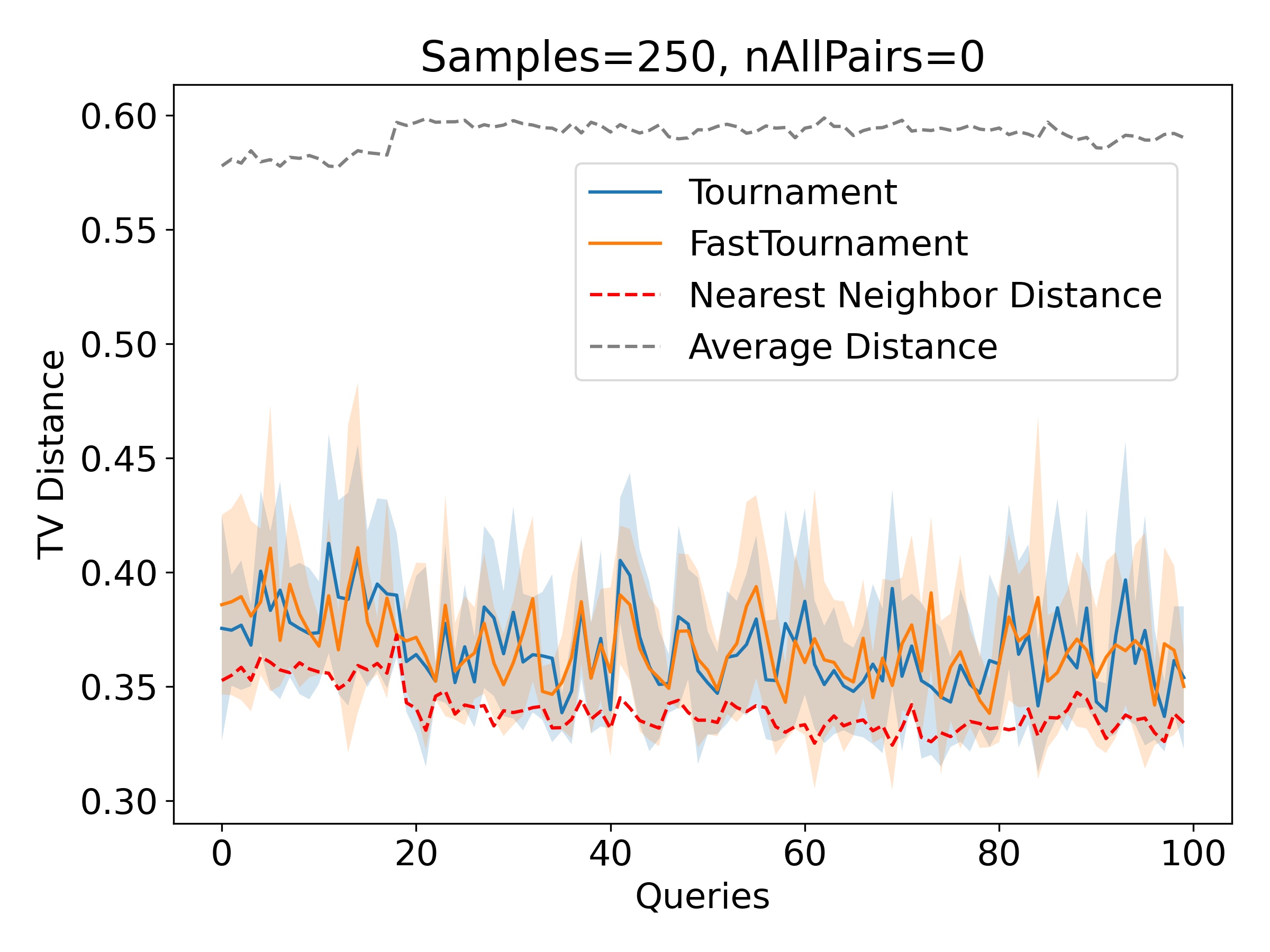}
     \end{subfigure}
\caption{Networking experiments with $\allpairs = 0$. The base tournament uses $5\times$ more operations than the fast tournament.}
\label{fig:ip0}
\end{figure}

\begin{figure}[t]
\centering
     \begin{subfigure}[b]{0.75\columnwidth}
         \centering
         \includegraphics[width=\textwidth]{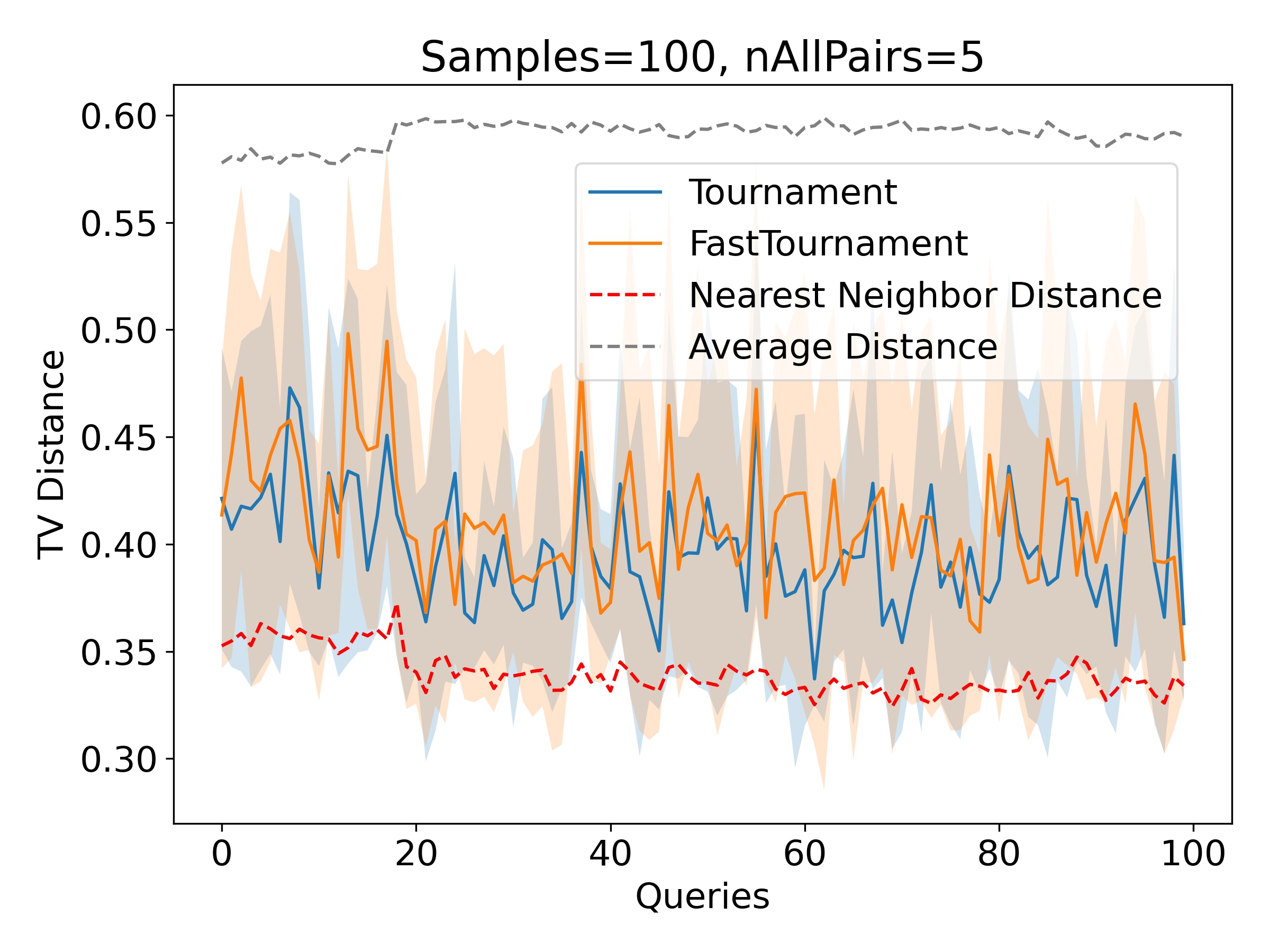}
     \end{subfigure}

     \begin{subfigure}[b]{0.75\columnwidth}
         \centering
         \includegraphics[width=\textwidth]{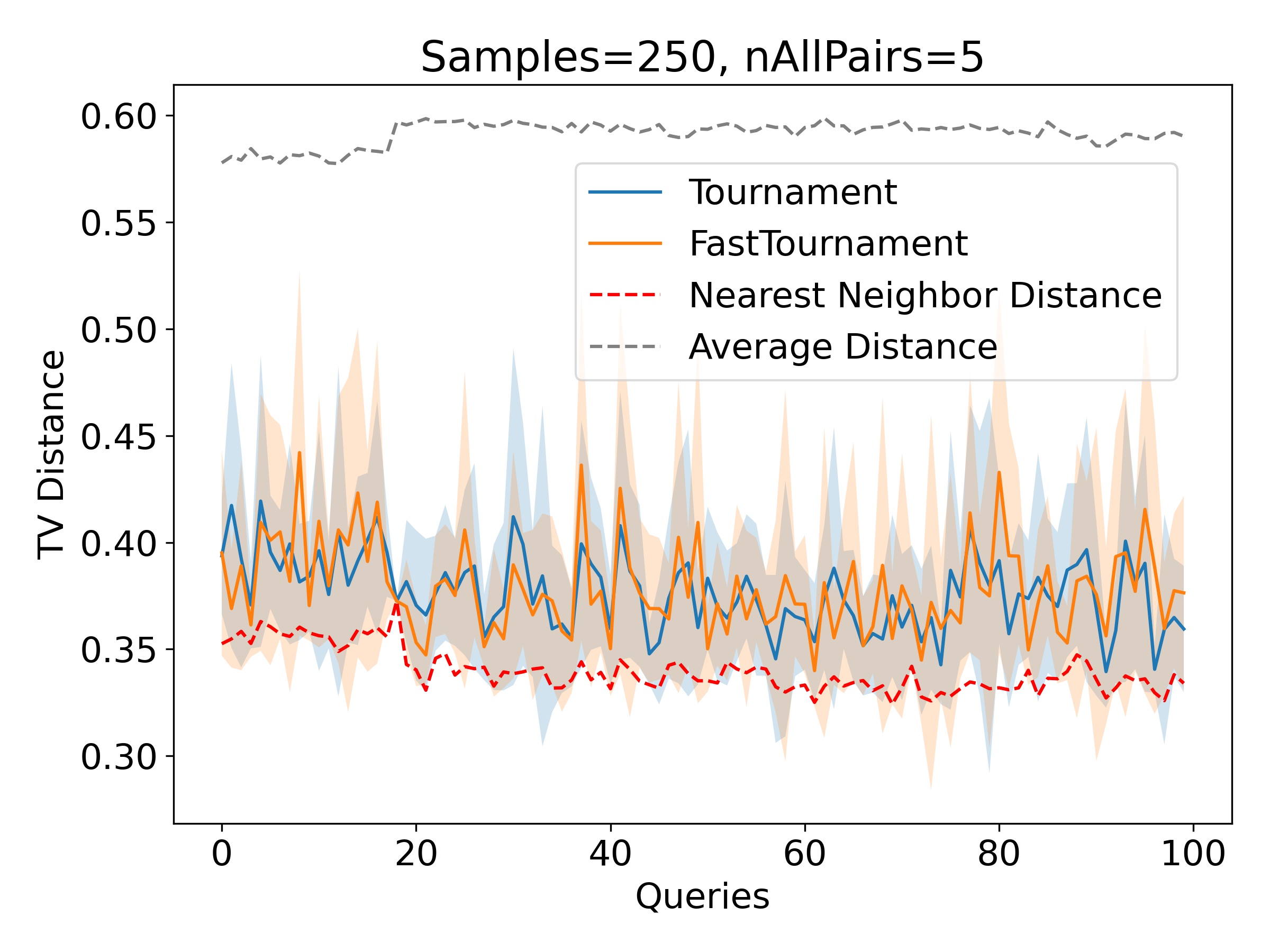}
     \end{subfigure}
\caption{Networking experiments with $\allpairs = 5$. The base tournament uses $2.4\times$ more operations than the fast tournament.}
\label{fig:ip5}
\end{figure}

\subsection{Networking Experiments}
\paragraph{Setup}
The underlying network data we use comes from the CAIDA Anonymized Internet Trace internet traffic dataset\footnote{From CAIDA internet traces 2019, \url{https://www.caida.org/catalog/datasets/monitors/passive-equinix-nyc/}}, which is IP traffic data collected at a backbone link of a Tier1 ISP in a data center in NYC. Within each minute, there are approximately $\approx 3.5 \cdot 10^7$ packets recorded.

We split $7$ minutes of the IP data into $2{,}148$ chunks, each representing $\approx 170$ms and approximately $10^{5}$ packets.
For each chunk, we construct a distribution corresponding to the empirical distribution over source IP addresses in that chunk.
The support sizes of the chunks (and therefore of the distributions) are approximately $45{,}000$.
The goal of approximate nearest neighbor search within this context is to identify similar traffic patterns to the current chunk of data within past data.
Further, algorithms for density estimation allow this computation to be done on subsampled traffic data, which is a common practice in networking to make data acquisition feasible~\cite{cisco}. 

For a series of 100 query chunks, we use the prior $2{,}048$ chunks as the set of distributions to search over.
We test the algorithms in two parameter regimes: $100$ samples and $\fast=10$ as well as $250$ samples and $\fast=25$. In both regimes, we test with $\allpairs=0$ and $\allpairs=5$. To measure performance, we report the total variation distance of the distribution returned by the tournament algorithms as well as the true nearest neighbor distance and average distance across all $2{,}048$ distributions.
Results are averaged over $10$ trials for each of the $100$ queries and one standard deviation is shaded.

\paragraph{Results}
In Figures \ref{fig:ip0} and \ref{fig:ip5}, we plot the performance of the base and fast tournaments on the networking data across samples $100$ and $250$ and for $\allpairs$ set to 0 and 5, respectively.
In all parameter settings, across the 100 queries, the fast tournament returns a distribution within roughly the same distance as the base tournament.
The fast tournament uses $5\times$ (when $\allpairs=0$) or $2.4\times$ (when $\allpairs=5$) less number of operations. Interestingly, neither algorithm performs better with the inclusion of an all-pairs comparison at the end of the tournament.
Even with the relatively small sampling rate compared to a domain of all possible IPs and support size of $45$k, the tournament algorithms are able to recover distributions close to the nearest neighbor distance. Comparing the results with $100$ and $250$ samples, the distance to the distribution returned by the algorithms moderately improves and the variance reduces considerably with more samples.

\paragraph{Experiment summary}
In both the synthetic and real-world experiments, the tournament algorithms perform well in recovering close distributions using few samples. The most important factor in the algorithms' performance is the number of total samples. On the other hand, the fast tournament is able to use very limited sample sizes for earlier rounds of the tournament in order to save up to \textbf{5x} on computational cost while retaining essentially the same performance to the base tournament which uses all samples at all steps. According to the theoretical results, this computational gap will only increase for larger $k$.

\section{Conclusion}
We introduce the question of sublinear time density estimation, a natural generalization of nearest neighbor search to discrete distributions. We obtain the first algorithm with sublinear sample complexity and query time and with polynomial preprocessing, in the proper case. In the improper case, we improve upon prior results of~\cite{acharya2018maximum} to obtain a linear time algorithm with optimal sample complexity. Our work raises a number of interesting follow up questions: To what extent can our upper bounds of query and sample complexity be improved? What are the computational-statistical tradeoffs between the sample complexity and query time? While we studied discrete distributions under total variation distance, it is also interesting to ask if we can obtain efficient retrieval algorithms for other distances for distributions, discrete or continuous.

\section*{Acknowledgements} Anders Aamand is supported by
DFF-International Postdoc Grant 0164-00022B from the Independent Research Fund Denmark. 
Research supported in part by NSF (CCF-2008733), and ONR (N00014-22-1-2713) grants. 
Justin Y. Chen is supported by  MathWorks Engineering Fellowship, GIST-MIT Research Collaboration grant, NSF award CCF-2006798, and Simons Investigator Award.
Justin Y. Chen, Shyam Narayanan, and Sandeep Silwal are supported by an NSF Graduate Research Fellowship under Grant No. 1745302.
Piotr Indyk was supported by the NSF TRIPODS program (award DMS-2022448), Simons Investigator Award and MIT-IBM Watson AI Lab.
Shyam Narayanan is supported by a Google PhD fellowship.

\bibliography{bib}
\bibliographystyle{icml2023}

\newpage
\appendix
\onecolumn

\section{Omitted Proofs of Section \ref{sec:motivation}}\label{App:omitted-proofs-motivation}

\lightbad*
\begin{proof}
Let $p=(p(1),\dots, p(n))$ where each $p(i)=1/n$. Let $x(i)\sim\text{Bin}(s,1/n)$ denote the number of times item $i$ is sampled, so that $\hat p(i)=x(i)/s$.
Fix a subset $A\subseteq [n]$ with $|A|\leq n/2$ and define $q_A=(q_A(1),\dots,q_A(n))$ by $q_A(i)=2/n$ if $i\in A$ and $q_A(i)=0$ otherwise. Note that $\|p-q_A\|_1=1$. Let $\ell_1=|\{i\in A\mid x(i)\geq 1\}|$ and $\ell_2=|\{i\in [n]\setminus A\mid x(i)\geq 1\}|$. If $s\leq n/2$, then for all $i$ such that $x(i)\geq 1$ it holds that
\[
|p(i)-\hat p(i)|=\begin{cases}
x(i)/s-p(i),& x(i)\geq 1  \\
p(i), & x(i)=0,
\end{cases}
\]
and a similar equation holds for $|q_A(i)-\hat p(i)|$. Using this, simple calculations show that
$\|\hat p - p\|_1=2-\frac{\ell_1+\ell_2}{n}$ and $\|\hat p - q_A\|_1=2-\frac{4\ell_1}{n}$. In particular, $\|\hat p - q_A\|_1-\|\hat p - p\|_1=\frac{2}{n}(\ell_1-\ell_2)$. By symmetry, if $A$ is sampled at random (and in particular $A$ and $A^c$ is sampled with the same probability), the events $\|\hat p - q_A\|_1>\|\hat p - p\|_1$ and $\|\hat p - q_A\|_1<\|\hat p - p\|_1$ occur with exactly the same probability. Moreover, conditioning on the $x(i)$'s, it is simple to check that regardless of their values, $\ell_1\neq \ell_2$ with probability $\Omega(1)$. It follows that the probability that $\|\hat p - p\|_1\leq \|\hat p - q_A\|_1$ is at most $1-c$ for some constant $c>0$. Now pick the distributions $q_1,\dots,q_{2k}$ independently and uniformly at random by picking random $A_1,\dots, A_{2k}$ with $|A_i|=n/2$ and defining $q_i=q_{A_i}$ (note that we may have repetitions). By independence, the probability that $\|\hat p - p\|_1\leq \min \|\hat p - q_i\|_1$ is $\exp(-\Omega(k))$. If $k=O(n\log n)$, it moreover holds with the same high probability that $\{q_1,\dots,q_{2k}\}$ contains at least $2k$ different distributions and so the result follows. If on the other hand $k\geq Cn\log n$ for $C$ sufficiently large, then we can pick $k'=Cn\log n$ distributions $q_1,\dots,q_{k'}$ satisfying the statement of the theorem. The error probability is $P_0=\exp(-\Omega(k'))$ but there are at most $n^{n/2}$ ways to do the sampling from $p$. Therefore, if for at least one way of sampling from $p$ it was the case that $\|\hat p - p\|_1\leq \min_{i\in[k']} \|\hat p - q_i\|_1$, then this event would in fact occur with probability at least $n^{-n/2}$. This is a contradiction since $P_0=\exp(-\Omega(k'))<n^{-n/2}$ when $C$ is sufficiently large.
\end{proof}

\heavybad*
\begin{proof}[Proof (sketch).]
Suppose for simplicity that $n=2n_0+1$ is odd. 
Let $p=(p(1),\dots, p(n))$ be given by 
\[
p(i)=
\begin{cases}
\frac{1}{2}, &i=1, \\
\frac{1}{2n_0},& i=2,\dots,n_0+1, \\ 
0,& i=n_0+2,\dots,2n_0+1, \\ 
\end{cases}
\]
Simple concentration bounds show that with high probability, $\sum_{i=2}^{n}(X(i)-sp(i))^2\leq (1+o(1))s/2$.
Define next $q=(q(1),\dots, q(n))$ by 
\[
q(i)=
\begin{cases}
\frac{1}{2}+\frac{1}{\sqrt{s}}, &i=1, \\
0,& i=2,\dots,n_0+1, \\ 
\frac{\frac{1}{2}-\frac{1}{\sqrt{s}}}{n_0},& i=n_0+2,\dots,2n_0+1, \\ 
\end{cases}
\]
Again with the assumption that $s=o(n)$, standard concentration bounds show that $\sum_{i=2}^{n}(X(i)-sq(i))^2\leq (1+o(1))s/2$ with high probability. It follows from these observations that if $X(1)\geq s/2+\sqrt{s}$, then $\|p-\hat p\|_2>  \|q-\hat p\|_2$. By standard properties of the binomial distribution $B(s,1/2)$, this happens with probability $\Omega(1)$.  
\end{proof}

\section{Omitted Proofs of Section \ref{sec:heavy_proofs}}\label{app:heavy-lemmas}

\begin{remark}
    We remark that in the proofs of this and all subsequent sections, the notation $1/\text{poly}(n)$ or $1/\text{poly}(k)$ refers to quantities of the form $1/n^C$ or $1/k^C$ where we can choose $C$ to be an arbitrarily large constant.
\end{remark}

The following is a standard tail bound for Poisson distributions\footnote{see e.g. \url{https://math.stackexchange.com/questions/2434883/chernoff-style-bounds-for-poisson-distribution/2434922} }.
\begin{lemma}\label{lem:pois}
If $Y \sim \text{Poi}(\lambda)$, then for any $t > 0$,
$\mathbb{P}(|Y  - \lambda| \ge t) \le 2 \exp\left(- \frac{t^2}{2(\lambda + t)} \right)$.    
\end{lemma}

\Cref{lem:p_p_hat_l_infinity} below bounds the $\ell_{\infty}$ distance between $\hat{p}$, the empirical distribution, and $p$, the unknown distribution. 

\begin{restatable}{lemma}{phatlinfinity}\label{lem:p_p_hat_l_infinity}
   Let $\hat{p}$ denote the empirical distribution of Algorithm \ref{alg:main_algo}. With probability at least $1-1/\text{poly}(n)$, we have that for all $i \in [n]$,
\[ | \hat{p}(i) - p(i) | \le O\left( \max\left( \sqrt{\frac{p(i) \log n}s}, \frac{\log n}s \right) \right).\] 
\end{restatable}

\begin{proof}
Note that $s \cdot \hat{p}(i)$ is distributed as \text{Poi}$(s \cdot p(i))$. By setting $t = C' \cdot \max\left( \sqrt{s 
\cdot p(i) \log n}, \log n \right)$ in Lemma \ref{lem:pois} for a large enough constant $C'$, we see that the probability $s \cdot \hat{p}(i)$  deviates by $t$ is at most $1/\text{poly}(n)$, where we can make the degree of the polynomial arbitrarily large by increasing $C'$.
The proof is completed by dividing by $s$.
\end{proof}

\numheavyelemes*
\begin{proof}
There can be at most $1/\gamma = n^C$ heavy elements since each heavy element has probability mass at least $\gamma$, by Algorithm \ref{alg:heavy_light_decomposition}. Let $H$ be the set of heavy elements. We know that any two distributions $w$ and $w' \in S_j$ are at most $O((\log n) \cdot (\log \log n))/\sqrt{n}$ apart in $\ell_{\infty}$. Thus, the $\ell_1$ difference restricted to the heavy elements is at most 
$ n^C \cdot O\left( \frac{(\log n) \cdot (\log \log n)}{\sqrt{n}} \right) = O\left(\frac{(\log n) \log \log n}{n^{1/2 - C}} \right)$.
\end{proof}

\optingroup*
\begin{proof}
We know $\|p - v^* \|_{\infty} \le \|p- v^*\|_2 \le \frac{1}{2 \sqrt{n}}$
by our assumption on $v^*$ given in Section \ref{sec:prelim}. Furthermore, Lemma \ref{lem:p_p_hat_l_infinity} implies that
$\|p - \hat{p}\|_{\infty} \le O\left( \frac{(\log n)^{3/4}}{\sqrt{n}} \right)$
so by adjusting constants, it follows from the triangle inequality that 
$\|\hat{p}  - v^*\|_{\infty} \le O\left( \frac{(\log n)^{3/4}}{\sqrt{n}} \right)$.
By definition of $v^{\infty}$, the distribution returned by the $\ell_{\infty}$ data structure in line $6$ of Algorithm \ref{alg:main_algo}, we know that $\|\hat{p} - v^{\infty}\|_{\infty} \le c\|\hat{p} - v^*\|_{\infty}$, where $c = O(\log(n) \cdot \log \log n)$ set in line $6$ of Algorithm \ref{alg:main_algo}. Hence,
\begin{align*}
    \|v^{\infty} - v^*\|_{\infty} &\le \|\hat{p} - v^{\infty}\|_{\infty} + \|\hat{p} - v^*\|_{\infty} \\
    &\le (c+1) \|\hat{p} - v^*\|_{\infty} \\
    &\le O\left(\frac{(\log n)^{1.75} \cdot (\log \log n)}{\sqrt{n}}\right),
\end{align*}
so $v^*$ must be in $S_{\infty}$ from line $6$ of Algorithm \ref{alg:preprocessing}.
\end{proof}

\pcloseloneheavy*
\begin{proof}
Take any $w \in S_{\infty}$. We know that
\begin{align*}
   &\|\hat{p} -w\|_{\infty} \le \|\hat{p} - v^{\infty} \|_{\infty} + \|v^{\infty} - w\|_{\infty}\\
   &\le \|\hat{p} - v^{\infty}\|_{\infty} + O\left(\frac{(\log n)^2 \cdot (\log \log n)}{\sqrt{n}} \right),
\end{align*}
where the last inequality holds from the construction of $S_{\infty}$ in Algorithm \ref{alg:preprocessing}.
Since $\|\hat{p} - v^{\infty} \|_{\infty} \le  c\|\hat{p} - v^* \|_{\infty}$, for $c$ defined in line 6 of Algorithm \ref{alg:main_algo}, by definition,
\begin{align*}
&\|p-w\|_{\infty} \le \|\hat{p} - w\|_{\infty} + \|p - \hat{p}\|_{\infty}\\
&\le c\|\hat{p} - v^*\|_{\infty} + \|p - \hat{p}\|_{\infty} + O\left(\frac{(\log n) \cdot (\log \log n)}{\sqrt{n}} \right).
\end{align*}

The proof of Lemma \ref{lem:opt_in_group} implies that $\|p - \hat{p}\|_{\infty}, \|\hat{p} - v^*\|_{\infty} \le O((\log n)^{3/4}/\sqrt{n})$, so by adjusting constant factors, the first inequality in the lemma statement holds, as $O\left((\log n)^2 \log \log n/\sqrt{n} \right)$ is the dominant term.

The second inequality in the lemma statement holds by an identical reasoning as used in the proof of Lemma \ref{lem:num_heavy_elems}, as there are at most $n^C$ heavy coordinates.
\end{proof}

\section{Omitted Proofs of Section \ref{sec:light_proofs}}\label{app:light-lemmas}

Recall that $T = \sum_{i \in L } p(i)$.
\lightEV*
\begin{proof}
Note that the light elements refer to elements in $L$ defined in Algorithm \ref{alg:main_algo}.
Set $X(i) := s \cdot \hat{p}'_L(i)$ and recall that $X(i) \sim \text{Pois}(s \cdot p_L(i))$.
We have
\begin{align*}
    &\mathbb{E}[\|s \cdot \hat{p}'_L-s \cdot v_L\|_2^2] \\
    &= \sum_{i=1}^n \mathbb{E}[ X(i)^2 + s^2v_L(i)^2 - 2sX(i)v_L(i)] \\
    &= \sum_{i=1}^n sp_L(i) + \sum_{i=1}^n \big (s^2p_L(i)^2 - 2s^2p_L(i)v_L(i) + s^2 v_L(i)^2  \big )\\
    &= sT + s^2 \|p_L-v_L\|_2^2. \qedhere
\end{align*}
\end{proof}

We now bound the variance. The following lemma bounds the variance one term at a time. 

\begin{restatable}{lemma}{varianceOne}\label{lem:variance_one}
Let $X(i) = s \cdot \hat{p}'_L(i)$. We have 
$\textup{Var}[(X(i)-sv_L(i))^2] \le 4sp_L(i) \cdot (sp_L(i)-sv_L(i))^2 + 6(sp_L(i))^2 + sp_L(i)$.
\end{restatable}

\begin{proof}
Define the auxiliary variables $y := sp_L(i)$ and $z:= sv_L(i)$. We have
\begin{align*}
    \text{Var}[(X(i)-sv_L(i))^2] &= \mathbb{E}[(X(i)-sv_L(i))^4] - \mathbb{E}[(X(i)-sv_L(i))^2]^2 \\
    &= \mathbb{E}[(X(i)-sv_L(i))^4] - ( s^2 p_L(i)^2 - 2s^2p_L(i)v_L(i) + s^2v_L(i)^2 + sp_L(i) )^2 \\
    &= \mathbb{E}[(X(i)-sv_L(i))^4] - ( y+ y^2 - 2yz + z^2 )^2.
\end{align*}
 Since $X(i)$ is a Poisson random variable with parameter $y$, we know that
\begin{align*}
    \mathbb{E}[X(i)^3] &= y + 3y^2+y^3, \\
    \mathbb{E}[X(i)^4] &= y + 7y^2 + 6y^3 + y^4.
\end{align*}
By expanding, we have
\begin{align*}
    \mathbb{E}[(X(i)-sv(i))^4] &= \mathbb{E}[X(i)^4] -4 \mathbb{E}[X(i)^3]z + 6 \mathbb{E}[X(i)^2]z^2 - 4 \mathbb{E}[X(i)] z^3 + z^4 \\
    &= (y + 7y^2 + 6y^3 + y^4) -4(y + 3y^2+y^3)z + 6(y+y^2)z^2 - 4yz^3 + z^4.
\end{align*}
Putting everything together gives us
\begin{align*}
&\text{Var}[(X(i)-sv_L(i))^2] = \mathbb{E}[(X(i)-sv_L(i))^4] - ( y+ y^2 - 2yz + z^2 )^2 \\
&= \big( (y+7y^2+6y^3+y^4) - 4(y+3y^2+y^3)z + 6(y+y^2)z^2-4yz^3+z^4\big)-(y+y^2-2yz+z^2)^2\\
&= 4y^3 - 8y^2z + 4yz^2 + 6y^2 + y -4yz \\
    &= 4y(y-z)^2+6y^2 + y - 4yz \\
    &\le  4y(y-z)^2+6y^2 + y \\
    &= 4sp_L(i) \cdot (sp_L(i)-sv_L(i))^2 + 6(sp_L(i))^2 + sp_L(i). \qedhere
\end{align*}
\end{proof}

The following lemma bounds the total variance.
\begin{lemma}\label{lem:variance}
$\textup{Var}[\|s \cdot \hat{p}'_L - s \cdot v_L\|_2^2]\le 4s^3 \|p_L\|_2 \|p_L-v_L\|_2^2 + 6s^2\|p_L\|_2^2 + sT$.
\end{lemma}
\begin{proof}
Summing the result of Lemma \ref{lem:variance_one} for all coordinates $i$ gives us
\begin{align*}
&\textup{Var}[\|s \cdot \hat{p}'_L - s \cdot v_L\|_2^2] \\
&\le \sum_{i=1}^n 4sp_L(i) \cdot (sp_L(i) - sv_L(i))^2 + 6(sp_L(i))^2 + sp_L(i)\\
&=  \sum_{i=1}^n 4s^3p_L(i) \cdot (p_L(i) - v_L(i))^2 + 6s^2 \|p_L\|_2^2 + sT \\
    &\le 4s^3 \|p_L\|_2 \|p_L-v_L\|_2^2 + 6s^2\|p_L\|_2^2 + sT. \qedhere
\end{align*}
\end{proof}

\lightconcentrationone*
\begin{proof}
By Chebyshev's inequality and Lemma \ref{lem:variance},
\begin{align*}
    &\mathbb{P}[|Z_1 - \mathbb{E}[Z_1]| \ge t] \le \frac{\text{Var}[Z_1]}{t^2}\\
    &\lesssim \frac{s^3 \|p_L\|_2 \|p_L-v_L\|_2^2}{s^4\eps^4/n^2} + \frac{s^2 \|p_L\|_2^2}{s^4 \eps^4/n^2} + \frac{s}{s^4 \eps^4/n^2} \\
    &\le \frac{n\|p_L\|_2}{s \eps^2} + \left(\frac{n \|p_L\|_2}{s \eps^2} \right)^2 + \frac{n^2}{s^3\eps^4}
    \le 0.01
\end{align*}
where the last inequality holds if $s \ge C\max\left(\frac{n \|p_L\|_2}{\eps^2}, \frac{n^{2/3}}{\eps^{4/3}}\right)$ for a sufficiently large constant $C$.
\end{proof}

We then consider the case where $\eps/\sqrt{n} \lesssim \|p_L-v_L\|_2$. As stated in the main body, we need to obtain a stronger concentration result since later we union bound over possibly $\Omega(k)$ different distributions which are in group $S_{\infty}$.

To obtain a stronger concentration result, we need the coordinates of both $\hat{p}_L$ and $v_L$ to be bounded. For each $v_i$, we know that $\|(v_i)_L\|_{\infty} \le 2\gamma$ by construction if we set $\gamma \ge \tilde{\Omega}(1/\sqrt{n})$. Lemma \ref{lem:p_p_hat_l_infinity} readily implies a similar statement for $\hat{p}_L$ (with high probability).

\begin{remark}\label{rem:tau_bound}
Therefore, in the following concentration bound, we utilize the fact that $\max_i \{ p_L(i), v_L(i) \} \le O(\gamma)$, which holds with high probability.
\end{remark}
The main tool we use is Bernstein's concentration inequality on $Z_1 = \|s \cdot \hat{p}'_L-s \cdot v_L\|_2^2$. Towards this, we first prove a bound on a variance like quantity.

\begin{lemma}\label{lem:bernstein}
With probability $1-1/\textup{poly}(k)$ (over the randomness in $\hat{p}'$), we have
\begin{align*}
 &\max_{i \in [n]} |(s \cdot \hat{p}'_L(i)  - s \cdot v_L(i))^2 - \E[(s \cdot \hat{p}'_L(i) - s \cdot v_L(i))^2]| \le O\left( \log k \cdot (s\gamma)^{1.5})\right).
\end{align*}

\end{lemma}
\begin{proof}
Define $X(i) = s \cdot \hat{p}'_L(i)$.
Recall each $X(i)$ is distributed as $\text{Pois}(s \cdot p_L(i))$. From Lemma \ref{lem:pois}, we have that with probability $1-1/\text{poly}(k)$,  $|X(i) - sp_L(i)| \le O\left(\max\left( \sqrt{sp_L(i) \log k}, \log k \right)\right)$ for \emph{all} $i$. Call this event $\mathcal{E}$.
Now
\begin{align*}
&(X(i) - s \cdot v_L(i))^2 - \E[(X(i) - s \cdot v_L(i))^2] \\
&= X(i)^2 - 2sv_L(i) X(i) + s^2 v_L(i)^2 - (s^2p_L(i)^2 - 2s^2p_L(i)v_L(i) + s^2 v_L(i)^2 + sp_L(i))\\
&= X(i)^2 - s^2p_L(i)^2 - sp_L(i) + 2sv_L(i)(sp_L(i) - X(i)).
\end{align*}

Thus conditioning on $\mathcal{E}$ and recalling Remark \ref{rem:tau_bound}, we have 
\[|X(i)^2 - s^2p_L(i)^2 - sp_L(i)| \le O((sp_L(i))^{1.5} \log k ) = O(  (s \gamma)^{1.5} \log k) \]
and 
\[|2sv_L(i)(sp_L(i) - X_L(i))| \le O(  s \gamma \cdot (\sqrt{s \gamma} + 1) \log k ) \]
for all $i$.
Altogether, we have that under $\mathcal{E}$,
\begin{align*}
 &\max_i |(X(i) - s \cdot v_L(i))^2 - \E[(X(i) - s \cdot v_L(i))^2]| \\
 &\le O\left(\log k \cdot (s^{1.5} \gamma^{1.5}\right)). \qedhere
\end{align*}

\end{proof}

\lightconcentrationtwo*
\begin{proof}
Define $X(i) = s \cdot \hat{p}'_L(i)$.
Recall Bernstein's inequality: for a random variable $R = \sum_i R_i$ where $R_i$ are independent, it states
\begin{equation}\label{eq:bernstein}
   \mathbb{P}(|R - \mathbb{E}[R]| \ge t) \le 2\exp\left( -\frac{t^2/2}{\text{Var}(R) + tM/3} \right)
\end{equation}
where $M$ is such that $|R_i - \E[R_i]| \le M$ with probability $1$ for every $i$. In our case, $Z_1 = \sum_i (X(i) - s \cdot v_L(i))^2$, so we must first bound the maximum deviation of \[|(X(i) - s \cdot v_L(i))^2 - \E[(X(i) - s \cdot v_L(i))^2]|\] for all $i$. To do so, we condition on the event $\mathcal{E}$ of Lemma \ref{lem:bernstein}. This gives
\[\mathbb{P}[| Z_1 - \mathbb{E}[Z_1] | \ge t]
=   \mathbb{P}[| Z_1 - \mathbb{E}[Z_1]| \ge t \mid \mathcal{E}] \cdot \mathbb{P}(\mathcal{E}) + \mathbb{P}[| Z_1 - \mathbb{E}[Z_1]| \ge t \mid \mathcal{E}^c] \cdot \mathbb{P}(\mathcal{E}^c).\]

We have
$\mathbb{P}[| Z_1 - \mathbb{E}[Z]| \ge t \mid \mathcal{E}^c] \cdot \mathbb{P}(\mathcal{E}^c) \le \mathbb{P}(\mathcal{E}^c) \le \frac{1}{\text{poly}(k)}$
so it suffices to bound $\mathbb{P}[| Z_1 - \mathbb{E}[Z_1]| \ge t \mid \mathcal{E}]$. We drop the conditioning $\mathcal{E}$ for simplicity.
Note that a similar analysis as above also shows that 
$|\mathbb{E}[Z_1] - \mathbb{E}[Z_1 \mid \mathcal{E}]| \le \frac{1}{\text{poly}(k)}$
which is a much smaller lower order term compared to $t$ so we ignore it for simplicity. A similar statement holds for the variance of $Z_1$. 

Now Bernstein's inequality gives us
\begin{align*}
&\mathbb{P}(|Z_1 - \mathbb{E}[Z_1]| \ge t) \lesssim \exp\left( -\frac{t^2/2}{4s^3\|p_L\|_2\|p_L-v_L\|_2^2 + 6s^2\|p_L\|_2^2 + s + tM/3} \right)
\end{align*}

where we have used Lemma \ref{lem:variance} to substitute in the variance. We can furthermore set $M = O\left(\log k \cdot (s\gamma)^{1.5}\right))$ from Lemma \ref{lem:bernstein}.
If we show that
\[\min\left(  \frac{t^2}{s^3\|p_L\|_2 \|p_L-v_L\|_2^2},  \frac{t^2}{s^2\|p_L\|_2^2},  \frac{t^2}{s},  \frac{t}{M} \right) \ge \Omega(\log k),\]
where the constants in the $\Omega(\log k)$ are sufficiently large,
then it follows that 
\[\frac{t^2/2}{4s^3\|p\|_2\|p_L-v_L\|_2^2 + 6s^2\|p_L\|_2^2 + s + tM/3} \ge \Omega(\log k) \]
and we get the desired concentration bound. Thus, we proceed to analyze each of the four fractions individually. Using $t =  \frac{s^2 \|p_L-v_L\|_2^2}{4} + \frac{s^2 \eps^2}{4n}$, we see that \begin{align*}
    \frac{t^2}{s^3\|p_L\|_2 \|p_L-v_L\|_2^2 }  &\ge \frac{s\|p_L-v_L\|_2^2}{\|p_L\|_2} \ge \frac{s \eps^2}{n \|p_L\|_2}, \\
    \frac{t^2}{s^2 \|p_L\|_2^2}  &\ge \frac{s^2\eps^4}{n^2\|p_L\|_2^2}, \\
    \frac{t^2}{s} &\ge \frac{s^3 \eps^4}{16n^2}.
\end{align*}
Thus, $s \ge \Omega(\max(n^{2/3}\log(k)^{2/3}/\eps^{2/3}, n\|p_L\|_2 \log k/\eps^2))$ is required. We now consider $t/M$ Recalling their values, we observe that
\begin{align*}
     \frac{t}{s^{1.5}\gamma^{1.5}} &\ge \frac{s^{0.5}\eps^2}{n \gamma^{1.5}}.
\end{align*}
(Note that $s \gamma \gg 1$ in our setting). If $s \ge \Omega(n \|p\|_2 \log k / \eps^2)$, then $s\eps^2/(n\gamma) = \Omega(\log k)$. To ensure $\frac{s^{0.5}\eps^2}{n \gamma^{1.5}} \ge \Omega(\log k)$,
 we need to satisfy $s = \Omega\left(\ \frac{n^2\gamma^{3}\log(k)^2}{\eps^4} \right).$
Recalling the value of $\gamma$, we can easily check that $s \ge \Omega\left(\ \frac{n^2\gamma^{3}\log(k)^2}{\eps^4} \right)$ implies all other lower bounds we have imposed on $s$ so far, proving the lemma.
\end{proof}

\section{Proof of Theorem \ref{thm:main}}\label{sec:main_thm_proof}

\mainthm*
\begin{proof}
Let $v^{\infty}$ be the output of line $6$ of Algorithm \ref{alg:main_algo} and $S_{\infty}$ be the group of $v^{\infty}$ from Algorithm \ref{alg:preprocessing}, the preprocessing step. Let $H$ and $L$ be the partition of $[n]$ into heavy and light elements induced by $v^{\infty}$.
From Lemma \ref{lem:p_close_l_1_heavy}, we know that $p$ is close to $\tilde{v}$ on $H$ as $\tilde{v} \in S_{\infty}$: $\|p - \tilde{v}\|_1 \le o(1)$. Thus it suffices to bound $\|p_L - \tilde{v}_L\|_1$.

Lemma \ref{lem:opt_in_group} states that $v^* \in S_{\infty}$ with probability $1-1/\text{poly}(n)$. Since we assumed that $\|p - v^*\|_2 \le \eps/(2 \sqrt{n})$, Lemma \ref{lem:light_concentration_1} implies that 
\begin{equation}\label{eq:proof1}
    \|s \cdot \hat{p}'_L - s \cdot v^*_L\|_2^2 \le sT + \frac{5s^2 \eps^2}{16n} 
\end{equation}

with probability $1-o(1)$ (note Lemma \ref{lem:light_concentration_1} is stated as 
holding with probability $99\%$ but we are picking a much larger value of $s$ than required by the lemma. Plugging into the Chebyshev inequality bound readily gives us that the failure probability is $o(1)$).

Now we claim that $\tilde{v}$ \emph{cannot} satisfy $\|\tilde{v}_L - p_L\|_2 \ge 0.99 \cdot \eps/\sqrt{n}$. Let us assume for the sake of contradiction that $\|\tilde{v}_L - p_L\|_2 \ge 0.99 \cdot \eps/\sqrt{n}$. Then from Lemma \ref{lem:light_EV},
\[\mathbb{E}[ \|s \cdot \hat{p}'_L - s \cdot \tilde{v}_L\|_2^2] = sT + s^2 \|p_L - \tilde{v}_L\|_2^2,\] and Lemma \ref{lem:light_concentration_2} implies that with probability at least $1-1/\text{poly}(k)$, we have
\begin{equation}\label{eq:proof2}
 \|s \cdot \hat{p}'_L - s \cdot \tilde{v}_L\|_2^2
    \ge   s T + .75s^2\|p_L - \tilde{v}_L \|_2^2 - 0.25 \cdot  \frac{s^2\eps^2}n  \ge sT + 0.48 \cdot \frac{s^2 \eps^2}{n},
\end{equation}
where the last inequality follows from the assumption that $\|p_L - \tilde{v}\|_2^2 \ge (.99)^2 \eps^2/n$.
However, the ratio of the quantities on the right hand side of \eqref{eq:proof2} and \eqref{eq:proof1} is at least
\[\frac{sT + 0.48 \cdot \frac{s^2\eps^2}{n}}{sT +  \frac{5s^2 \eps^2}{16n}} \ge  1 + \frac{z/2}{T+z}\]
for $z = 5s\eps^2/(16n)$. Now if $T \le z$ then $z/2/(T+z) \ge 1/4$ so the ratio is at least $1.25$. Otherwise, if $T > z$, we always know that $T \le 1$ so $T+z \le 2$ and hence the ratio is at least 
$1+\frac{z}{4} \ge 1 + \frac{5s\eps^2}{64n}$.
Since $(1+x)^{1/2} \ge 1+x/2.5$ for any $x \in [0,1]$, we have that (in our regime of $s$) the \emph{square root} of the ratio is strictly larger than $1+s\eps^2/(32n)$,
by setting $s$ appropriately. However, this contradicts the $\ell_2$ nearest neighbor search guarantees set in line $10$ of Algorithm \ref{alg:main_algo}, since it means we return a $\tilde{v}$ which has the property that $\| \hat{p}'_L - \tilde{v}_L\|_2$ is larger than $\| \hat{p}'_L - v^*_L\|_2$ by a factor larger than $c$ set in line 10 (and $v^*$ was also considered by the $\ell_2$ data structure).

Therefore, we must have $\|\tilde{v}_L - p_L\|_2 \le 0.99 \cdot \eps/\sqrt{n}$ which implies that $\|\tilde{v}_L - p_L\|_1 \le 0.99 \cdot \eps$. Since we already know that 
$\|\tilde{v}_H - p_H\|_1 \le o(1)$, we have that $\|\tilde{v} - p\|_1 \le (1+o(1)) 0.99 \cdot \eps < \eps$. This proves the approximation.

For the sample complexity, note that if we set $s = \Theta(n/(\log k)^{1/4})$ then $s$ is larger than the values required in Lemmas \ref{lem:light_concentration_1} and \ref{lem:light_concentration_2}. The prepossessing time follows from Theorem \ref{thm:l_infinity_nns}. Finally, the query time follows from \ref{thm:l2_nns} by plugging in the value of $c$ from line 10 of Algorithm \ref{alg:main_algo} (the $\ell_2$ NNS approximation parameter) into statement of Theorem \ref{thm:l2_nns} and noting that all $\log(k)$ factors can be absorbed into the exponent as $k^{\Theta(1/(\log k)^{1/4})} \gg \text{poly}(\log k)$. Note that we get an extra $\tilde{O}(n)$ runtime from also querying an $\ell_{\infty}$ NNS datastructure as well in step $6$ of Algorithm \ref{alg:main_algo}.
Finally, we remark that an alternative expression for the approximation guarantee is $\|p - \tilde{v}\|_1 \le \|p - v^*\|_1 + \eps$.
This completes the proof.
\end{proof}

\section{Omitted Proofs of Section \ref{sec:fast-tournament}}\label{sec:fast_tournament_proof}

\fasttournament*
\begin{proof}
The statement about the number of samples is clear. Indeed, for the iterative part of the algorithm, we use a total of $s_{\lg k}=\frac{10\log(k^2/\delta)}{\eps^2}$ samples and for the final part, we use a further $\frac{10\log (|\mathcal{C}|^2/\delta)}{\eps^2}\leq \frac{10\log (k^2/\delta)}{\eps^2}$ samples. Regarding the running time, we already argued above that the time used for the knockout tournament is $O\left(\frac{k}{\eps^2}\log \frac{1}{\delta}\right)$. Further, since $|\mathcal{C}|\leq k^{1/3}\log k$, running the Scheffe test on all pairs of distributions in $\mathcal{C}$ takes time $O\left(sk^{2/3}\log^2 k\right)=O\left(\frac{k}{\eps^2}\log \frac{1}{\delta}\right)$.

For the bound on the quality of the estimate $\hat v$, we define $v^*=\argmin_{v\in \mathcal{V}} \|p- v\|_1$ and let  $\mathcal{W}_0=\{v\in \mathcal{V}: \|p- v\|_1\leq 3 \|p- v^*\|_1+\eps\}$ and $\mathcal{W}_1=\mathcal{V}\setminus \mathcal{W}_0$. We first want to argue that with probability $1-O(\delta)$, some element of $\mathcal{W}_0$ gets added to $\mathcal{C}$.
For $i=1,\dots,\lg k$, we let $T_i$ be the set of $2^i-1$ distributions that $v^*$ could possibly be paired with in step 2. of the algorithm during the  first $i$ rounds of the tournament (this is the set of $2^i-1$ distributions in the subtree of the tournament tree rooted $i$ steps above $v^*$). Let $A_i$ denote the event that there exists a distribution $v'$ in $T_i\cap \mathcal{W}_1$ such that $v^*$ loses the Scheffe test to $v'$ using the sample $\mathcal{S}_i$. By the choice of $s_i$ and the bound in~\eqref{eq:scheffe-estimate}, we have that
\[
\Pr[A_i]\leq |T_i\cap \mathcal{W}_1|\delta_i\leq \frac{\delta}{2^i},
\]
so $\Pr[\bigcup_i A_i]\leq \delta$. We now consider two cases: Suppose first that at some stage $i$, $\frac{|\mathcal{V}_i\cap \mathcal{W}_0|}{|\mathcal{V}_i|}\geq \frac{\log 1/\delta}{k^{1/3}}$. In this case, the probability that no element of $\mathcal{V}_i$ is added to $\mathcal{C}$ during step 1.~of iteration $i$ of the algorithm is at most $(1-\frac{\log 1/\delta}{k^{1/3}})^{k^{1/3}}\leq \delta$. If on the other hand, at each step $i$, $\frac{|\mathcal{V}_i\cap \mathcal{W}_0|}{|\mathcal{V}_i|}< \frac{\log 1/\delta}{k^{1/3}}$, then the probability that $v^*$ gets paired with an element of $\mathcal{W}_0$ in step $i$ (conditioned on it surviving the first $i-1$ steps) is at most $\frac{\log 1/\delta}{k^{1/3}}$. Note that if none of the events $A_i$ occurs and $v^*$ never gets paired with a distribution in $\mathcal{W}_0$, then $v^*$ will win the tournament, and thus get added to $\mathcal{C}$ in some step. Thus, by a union bound, the probability that no element of $\mathcal{W}_0$ gets added to $\mathcal{C}$ is at most
\[
\Pr\left[\bigcup_i A_i\right]+\delta+\frac{\log k\log 1/\delta}{k^{1/3}}\leq 2\delta+ \frac{(\log k)^2}{4k^{1/3}}=O(\delta),
\]
where we used $\delta\geq k^{-1/4}$ in the final step.

It was shown in~\cite{devroye2001combinatorial} that running the Scheffe test on all pairs of distributions in a set $\mathcal{C}$ using a sample of size $\frac{10\log \left(\binom{|\mathcal{C}|}{2}/\delta\right)}{\eps^2}$ and outputting the one $\hat v$ with the most wins, we have with probability at least $1-\delta$ that
\[
\|p-\hat v\|_1 \le 9 \cdot \min_{v\in \mathcal{C}} \|p- v\|_1 + 4\eps. 
\]
Since $\mathcal{C}$ contains a distribution of $\mathcal{W}_0$ with probability $1-O(\delta)$ we obtain that with probability $1-O(\delta)$,
\begin{align*}
&\|p-\hat v\|_1 \le 9 \cdot \min_{v\in \mathcal{C}} \|p- v\|_1 + 4\eps \\
&\le 9 \cdot (3 \|p- v^*\|_1+\eps) + 4\eps = 27 \|p- v^*\|_1+13\eps,
\end{align*}

as desired.
\end{proof}

\begin{remark}
We can also handle error probabilities $\delta<k^{-1/4}$. In this case, we set $\delta_0=k^{-1/4}$ and run the iterative part of the algorithm above $\ell=\frac{\log 1/\delta}{\log 1/\delta_0}$ times to obtain a candidate set of distributions $\mathcal{C}\subseteq \mathcal{V}$ ($\ell$ times larger than before) which contains an element of $\mathcal{W}_0$ with probability $1-O(\delta)$. We then again run the complete tournament on $\mathcal{C}$ with a sample of size $\frac{10\log \left(\binom{|\mathcal{C}|}{2}/\delta\right)}{\eps^2}$. The bound on the quality of the estimate, follows as above. Moreover, easy calculations show that the number of samples used is still $O(\frac{\log 1/\delta}{\eps^2})$ and that the total running time is 
\[
O\left(\frac{1}{\eps^2}\log \frac{1}{\delta}\left( k+\ell^2k^{2/3}\log^2k\right)\right),
\]
which is $O\left(\frac{k}{\eps^2}\log \frac{1}{\delta}\right)$ except for extremely small $\delta$.
\end{remark}
\begin{remark}
In a slightly different model, we are not given the distributions of $\mathcal{V}$ explicitly but can only access them through sampling. In this model, we can skip the preprocessing step and instead estimate the probabilities $v_i(S)$ and $v_j(S)$ through sampling whenever we run the Scheffe test for distributions $v_i$ and $v_j$. This comes at the cost of a larger value of the constant $C$. When running the algorithm above in this model, we obtain the same reduction of the running time and with no preprocessing.
\end{remark}

\end{document}